\documentclass[12pt]{article}
\usepackage{amsmath}
\usepackage{amssymb}
\usepackage{epsfig}
\usepackage{stmaryrd}
\usepackage{amsthm}
\usepackage{bm}

\numberwithin{equation}{section}
\numberwithin{figure}{section}

\newtheorem{theorem}{Theorem}

\newtheorem{definition}{Definition}
\newtheorem{claim}{Claim}
\newtheorem*{claim2}{Claim}

\def\eq#1{(\ref{eq:#1})}

\def\half{\fraction{1}{2}}

\def\fraction#1#2{ { \textstyle \frac{#1}{#2} }}
\def\d{\partial}
\def\lineup{\!\!\!\!\!\!\!\!&&}

\def\llangle{\langle\!\langle}
\def\rrangle{\rangle\!\rangle}
\def\leftllangle{\left\langle\!\!\!\left\langle}
\def\rightrrangle{\right\rangle\!\!\!\right\rangle}
\def\Bigllangle{\Big\langle\!\!\Big\langle}
\def\Bigrrangle{\Big\rangle\!\!\Big\rangle}

\def\Psimp{{\Psi_\mathrm{simp}}}
\def\Psch{{\Psi_\mathrm{Sch}}}

\begin{document}

\begin{titlepage}

\begin{center}

\vskip 1.0cm {\large \bf{Exotic Universal Solutions in Cubic Superstring 
Field Theory}}
\\
\vskip 1.0cm

{\large Theodore Erler\footnote{Email: tchovi@gmail.com}}
\vskip 1.0cm

{\it {Institute of Physics of the ASCR, v.v.i.} \\
{Na Slovance 2, 182 21 Prague 8, Czech Republic}}

\vskip 1.0cm
{\bf Abstract}
\end{center}
\noindent We present a class of analytic solutions of cubic 
superstring field theory in the universal sector on a non-BPS D-brane. 
Computation of the action and gauge invariant overlap reveal that the 
solutions carry half the tension of a non-BPS D-brane. However, the solutions
do not satisfy the reality condition. In fact, they display an 
intriguing topological structure: We find evidence that 
conjugation of the solutions is equivalent to a gauge transformation
that cannot be continuously deformed to the identity.  
\noindent 
\medskip

\end{titlepage}

\newpage

\baselineskip=18pt

\tableofcontents

\section{Introduction}

There are two interesting ways to formulate the field equations of an open 
NS superstring. The first comes from Berkovits' nonpolynomial string field 
theory\cite{Berkovits}, and involves a ghost and picture number $0$ 
string field $\Phi$ in the large Hilbert space subject to the equations of 
motion
\begin{equation}\eta_0(e^{-\Phi}Q e^\Phi)=0.\label{eq:berkovits}\end{equation}
The second comes from cubic superstring field theory\cite{PTY,AZ}, and 
involves a ghost number $1$, picture number $0$ string field $\Psi$ in the 
small Hilbert space subject to the equations of motion 
\begin{equation}Q\Psi+\Psi^2 = 0. \label{eq:pty}
\end{equation}
The cubic equations of motion are simpler, in that they are polynomial and 
directly analogous to the field equations for the open bosonic 
string\cite{Witten}, but suffer 
from the disadvantage that they are difficult to derive from a completely 
reliable action\footnote{To evaluate the energy
in this paper, we will use the action originally proposed by Preitschopf, 
Thorne, and Yost\cite{PTY}:
\begin{equation}S=\frac{1}{2}\llangle \Psi Q\Psi\rrangle+\frac{1}{3}\llangle
\Psi^3\rrangle.\end{equation}
The bracket $\llangle\cdot\rrangle$ is defined using the Witten vertex
with a midpoint insertion
\begin{equation}Y_{-2}=Y(i)\tilde{Y}(i),\ \ \ Y(z)=-\partial\xi e^{-2\phi}c(z).
\label{eq:Ym2}\end{equation} See appendix \ref{app:conventions}. 
The problems with this 
action are well-known, including difficulties with the convergence of level 
truncation\cite{Russians_level,Ohmori,Ohmori_rev,Raeymaekers}, complications 
with gauge fixing and perturbation theory\cite{PTY}, problems with the 
Ramond sector\cite{equivalence_Ramond,Kroyter_philosophy}, 
and the existence of a singular kernel for the bracket\cite{Berkovits_critic}. 
Recently there has been some interest in finding a more suitable 
action\cite{Kroyter_philosophy,democratic,BS}, though the success of these 
proposals remains unclear.}. Nevertheless
the Berkovits and cubic equations are known to be perturbatively 
equivalent\cite{equivalence}, and nonperturbatively any Berkovits solution 
generates a cubic solution via the equation
\begin{equation}\Psi=e^{-\Phi}Qe^{\Phi}.\label{eq:Berkcub}\end{equation}
However, the reverse is not true. The existence of a cubic solution $\Psi$
does not guarantee the existence of a Berkovits solution $e^{\Phi}$ satisfying
\eq{Berkcub}. For example, cubic superstring field theory has a ``tachyon
vacuum'' on a BPS D-brane\cite{Erler,Kroyter_myvac}. There is no evidence for 
such a solution in Berkovits' string field theory, either 
analytically\cite{Erler,ES} or numerically\cite{BSZ}. 

In this paper we present a new example of this phenomenon. We show that
the cubic equations of motion on a non-BPS D-brane possess an unexpected 
class of universal solutions which appear not to exist in Berkovits string 
field theory. The existence of these solutions is highly nontrivial, but
their physical interpretation is unknown. They possess a number of 
surprising properties which may be important for our understanding of string 
field theory:
\begin{itemize}
\item The solutions are not real. In fact, every solution appears to be 
related to its conjugate by a topologically nontrivial gauge transformation.
\item The solutions appear not to exist in Berkovits' string field theory.
\item If we ignore the reality condition and compute observables,  
the solutions turn out to carry half the tension of a non-BPS D-brane.
\end{itemize}
We will call them {\it half-brane solutions}, in accordance with their 
tension. The solutions are significant in that they appear to be the 
first examples of topological solutions in open string field theory. 
We hope that they can provide a deeper understanding of the 
topology of the string field algebra, with the ultimate goal of providing a 
``microscopic'' description of D-brane charges in the context of string 
field theory.

This paper is organized as follows. In section \ref{sec:solution} we construct
half-brane solutions by extending the wedge algebra to include 
generators of worldsheet supersymmetry. We attempt an analogous 
construction in Berkovits string field theory, and show that it fails. 
In section \ref{sec:reality} we prove that 
half-brane solutions do not satisfy the string field reality condition. We 
also show, within a controlled subalgebra of states, that every half-brane 
solution is related to its conjugate by a topologically nontrivial gauge 
transformation. In section \ref{sec:phantom} we discuss the regularization 
and phantom 
piece for the half-brane solution. The phantom term offers an interesting 
perspective on the nature of convergence in the wedge algebra, and suggests 
a more general technique for constructing states in the wedge 
algebra---including, possibly, projector states distinct from 
the sliver and identity string field. In section 
\ref{sec:observables} we calculate the action and closed string tadpole.  
We find highly nontrivial agreement between these observables, indicating 
that the solutions represent a state with half the tension of a non-BPS 
D-brane. We end with some discussion.

\section{Solution}
\label{sec:solution}

\subsection{Algebra} 
\label{subsec:algebra}

To begin we need to recall some facts about the algebra of string 
fields\footnote{In this paper we use the left handed convention for the star 
product\cite{simple}. Other standard sources for the 
superstring\cite{Ohmori,Raeymaekers,BSZ} use the right handed 
convention\cite{Okawa}, and there are some important 
sign differences in the GSO($-$) sector. See appendix 
\ref{app:conventions}.}
on a non-BPS D-brane. The algebra has two $\mathbb{Z}_2$ gradings: Grassmann 
parity $\epsilon$, which corresponds to the Grassmann parity of the vertex 
operator creating the string field; and worldsheet spinor 
number $F$, which tells us whether the field is in the GSO($+$) or GSO($-$) 
sector. Fields in the algebra are assigned internal Chan-Paton factors 
according to the table:
\begin{center}
\begin{tabular}{|c|c|c|c|}\hline
$\ \ \ \epsilon\ \ \ $ & $\ \ \ F\ \ \ $ & CP factor
\\ \hline
$0$ & $0$ & $\mathbb{I}$  \\ \hline
$1$ & $0$ & $\sigma_3$  \\ \hline
$0$ & $1$ & $\sigma_2$ \\ \hline
$1$ & $1$ & $\sigma_1$ \\ \hline
\end{tabular}.
\end{center}
The BRST charge $Q$ and the midpoint insertion $Y_{-2}$ both implicitly 
carry an internal CP factor of $\sigma_3$, and the 1-string vertex 
$\llangle \cdot\rrangle$ automatically contains a factor of $1/2$ times 
the trace over internal CP matrices. To keep track of signs when commuting
vertex operators and CP factors past each other, it is helpful to define what
we will call {\it effective Grassmann parity}
\begin{equation}E=\epsilon+F\ \ \ \ (\mathrm{mod}\ 2).\end{equation}
In particular, the star algebra has a natural graded 
commutator\footnote{This ``double bracket'' commutator should be 
distinguished from the graded commutator
$[\Psi,\Phi]=\Psi\Phi-(-1)^{E(\Psi)E(\Phi)}\Phi\Psi$ which emerges naturally
from the action, both in the infinitesimal gauge transformation and the 
kinetic operator around a nontrivial solution. The single bracket $[,]$ is 
only graded according to effective Grassmann parity.}
\begin{equation}\llbracket\Psi,\Phi\rrbracket = \Psi\Phi - 
(-1)^{E(\Psi)E(\Phi)+F(\Psi)F(\Phi)}\Phi\Psi,\end{equation}
where $\Psi,\Phi$ implicitly carry the appropriate CP factor. This 
suggests that the star product on a non-BPS brane has a structure 
analogous to a product of matrices whose entries contain two mutually 
commuting types of Grassmann number, the first has a Grassmannality measured 
by $E$ and the second by $F$. However, only effective Grassmann parity 
enters into the string field theory 
axioms:
\begin{eqnarray}Q(\Psi\Phi)\lineup = (Q\Psi)\Phi+(-1)^{E(\Psi)}\Psi(Q\Phi),
\nonumber\\
\llangle \Psi\Phi\rrangle \lineup= (-1)^{E(\Psi)E(\Phi)}
\llangle \Phi\Psi \rrangle.\end{eqnarray}
In particular, the physical string field $\Psi$ on a non-BPS D-brane 
must be {\it effective} Grassmann odd. For example, the 
tachyon vertex operator $\gamma(0)$ is Grassmann even in the traditional 
sense, but since it carries worldsheet spinor number, it counts 
as ``effectively'' Grassmann odd.

\begin{table}[t]
\begin{center}
\begin{tabular}{|c|c|c|c|c|c|c|}\hline
& ${\mathrm{ghost} \atop \mathrm{number}}$ 
& ${\mathrm{effective}\atop \mathrm{Grassmann\ parity}}$ 
& ${\mathrm{worldsheet}\atop \mathrm{spinor\ number}}$ 
& ${\mathrm{scaling}\atop\mathrm{dimension}}$ 
& reality 
& twist \\ \hline
$K$ & $0$ & $0$ & $0$ & $1$ & real & $1$ \\ \hline
$B$ & $-1$ & $1$ & $0$ & $1$ & real & $1$ \\ \hline
$c$ & $1$ & $1$ & $0$ & $-1$ & real & $-1$ \\ \hline
$G$ & $0$ & $0$ & $1$ & $\frac{1}{2}$ & real & $-i$ \\ \hline
$\gamma$ & $1$ & $1$ & $1$ & $-\frac{1}{2}$ & real & $-i$ \\ \hline
\end{tabular}
\end{center}
\caption{\label{tab:alg} Some important quantum numbers for the 
atomic fields. Scaling dimension refers to the eigenvalue of the 
field under the action of the operator
$\frac{1}{2}\mathcal{L}^-=\frac{1}{2}(\mathcal{L}_0-\mathcal{L}_0^\star)$.
Reality and twist refer to the eigenvalues of the fields under
reality and twist conjugation, defined in appendix 
\ref{app:conventions}. By ``real'' we mean that the fields have eigenvalue
$1$ under reality conjugation.}
\end{table}

With these preparations we are ready to give the algebraic setup for our 
solution. The solution is constructed by taking star products of 
four ``atomic'' string fields: 
\begin{equation}K,\ \ \ \ B,\ \ \ \ c,\ \ \ \ G.\ \end{equation}
The ghost number, effective Grassmann parity, and some other important 
quantum numbers of these fields are summarized in table \ref{tab:alg}. 
We can construct $K,B,c,G$ by acting certain operators on the 
identity string field $|I\rangle$: 
\begin{eqnarray}K \lineup = \mathbb{I}\otimes \mathcal{L}^+_L|I\rangle,
\ \ \ \ \ \ \ \ \ \ \ B = \sigma_3\otimes \mathcal{B}^+_L|I\rangle,\nonumber\\
c \lineup = \sigma_3\otimes\frac{1}{\pi}c(1)|I\rangle,\ \ \ \ \ \ 
G=\sigma_1\otimes \mathcal{G}_L|I\rangle.
\end{eqnarray}
The subscript $L$ above denotes taking the left half of the charges:
\begin{eqnarray}
\mathcal{L}^+ \lineup = \mathcal{L}_0+\mathcal{L}_0^\star, 
\ \ \ \ \ \ \ \ \ \ \ \ 
\mathcal{L}_0 = f_\mathcal{S}^{-1}\circ L_0,\nonumber\\
\mathcal{B}^+ \lineup = \mathcal{B}_0+\mathcal{B}_0^\star, 
\ \ \ \ \ \ \ \ \ \ \ \ 
\mathcal{B}_0 = f_\mathcal{S}^{-1}\circ b_0,\nonumber\\
\mathcal{G}\ \, \lineup = f_\mathcal{S}^{-1}\circ G_{-1/2},
\label{eq:SS_op}\end{eqnarray} 
where $f_\mathcal{S}^{-1}(z) = \tan\frac{\pi}{2}z$ is the inverse of the 
sliver conformal map\cite{Schnabl,RZO} and the star $^\star$ denotes BPZ 
conjugation. Another definition of these fields is given by mapping them
to operator insertions inside correlation functions on the cylinder:
\begin{eqnarray}
K\lineup \rightarrow \mathbb{I}
\int_{-i\infty}^{i\infty}\frac{dz}{2\pi i} T(z),
\ \ \ \ \ \ \,
B\rightarrow \sigma_3\int_{-i\infty}^{i\infty}\frac{dz}{2\pi i}b(z),
\nonumber\\
c\lineup \rightarrow \sigma_3 c(z),\ \ \ \ \ \ \ \ \ \ \ \ \ \ \ \ \ \ 
 G\rightarrow \sigma_1\int_{-i\infty}^{i\infty}\frac{dz}{2\pi i}G(z).
\end{eqnarray}
See \cite{Okawa,SSF1} and the appendix of \cite{simple} for an explanation
of how this mapping works. The essentially new ingredient in our algebraic 
setup is the string field $G$. It lives in the GSO($-$) sector, and corresponds
to a line integral insertion of the worldsheet supercurrent $G(z)$. 
Note that to define $G$ we need to ``split'' the operator $\mathcal{G}$ into 
left and right halves. Such splittings are potentially anomalous\cite{RZ}, 
but in this case the splitting appears to be regular (see appendix 
\ref{app:G} for more details). 

The fields $K,B,c,G$ freely generate a subalgebra of the open 
string star algebra subject to the relations
\begin{eqnarray}\lineup G^2 = K, \ \ \ \ \ Bc+cB = 1,\ \ \ \ \ B^2=c^2=0,
\nonumber\\
\lineup K,B,G\ \mathrm{mutually\  commute}.\label{eq:alg}\end{eqnarray}
$K$ generates the algebra of wedge states\cite{Okawa,RZ_wedge} in the sense
that any star-algebra power of the $SL(2,\mathbb{R})$ vacuum 
$\Omega=|0\rangle$ can be written $\Omega^\alpha = e^{-\alpha K}$. It is
useful to define the operators: 
\begin{equation}\partial = [K,\cdot],\ \ \ \ \ \ 
\delta = \llbracket G,\cdot\rrbracket. \end{equation}
In the cylinder coordinate frame, $\partial$ generates an infinitesimal 
worldsheet translation and $\delta$ generates a worldsheet 
supersymmetry variation. They are derivations of the star product:
\begin{eqnarray}
\partial(\Psi\Phi)\lineup = (\partial\Psi)\Phi + \Psi(\partial\Phi),\nonumber\\
\delta(\Psi\Phi)\lineup = (\delta\Psi)\Phi + (-1)^{F(\Psi)}\Psi(\delta\Phi).
\label{eq:dderiv}\end{eqnarray}
Since the supersymmetry variation of $c$ produces the $\gamma$ 
ghost, it is helpful to introduce the corresponding string field:
\begin{equation}\gamma = 
\sigma_2 \otimes \frac{1}{\sqrt{\pi}}\gamma(1)|I\rangle \ \rightarrow\ 
\sigma_2\gamma(z)= \sigma_2\eta e^\phi(z).\end{equation}
Then we have
\begin{eqnarray}\delta c\ = \lineup 2i\gamma,
\ \ \ \ \ \ \ \delta\gamma \ = -\frac{i}{2}\d c,\nonumber\\
\delta G = \lineup 2K,\ \ \ \ \ \ \ \delta K=0,\ \ \ \ \ \ \ \ \ \ \ \ \ \ 
\delta B=0.
\label{eq:dother}\end{eqnarray} 
Note that $\delta$ satisfies the supersymmetry algebra $\delta^2=\partial$.
Since $K$ is the worldsheet superpartner of $G$, together these fields 
generate a supersymmetric extension of the wedge algebra, which we will call
the {\it wedge superalgebra}. 

The algebra generated by $K,B,c,G$ is closed under the action of the 
BRST operator
\begin{equation}QK = 0,\ \ \ \ \ QB=K,\ \ \ \ \
Qc = cKc-\gamma^2,\ \ \ \ \ \  QG=0.\label{eq:BRST}\end{equation}
Therefore it makes sense to look for solutions to the cubic equations of
motion
\begin{equation}Q\Psi+\Psi^2=0\end{equation}
within this subalgebra.

\subsection{Half-Brane Solutions}
\label{subsec:solutions}

In this paper we study solutions of the form
\begin{equation}\Psi[f] = \left(c\frac{KB}{1-f}c + B\gamma^2\right)f,
\label{eq:sol1}\end{equation}
where $f$ is a string field in the wedge superalgebra. Multiplying and 
dividing by $\sqrt{f}$ gives a gauge equivalent solution
\begin{equation}\hat{\Psi}[f] = 
\sqrt{f}\left(c\frac{KB}{1-f}c+B\gamma^2\right)\sqrt{f},\label{eq:sol_real}
\end{equation}
which is twist symmetric\footnote{By twist symmetric, we mean that the GSO($+$)
and GSO($-$) components of the solution separately have definite eigenvalue 
under twist conjugation. In particular, the GSO($+$) component is made of 
states whose $L_0+1$ eigenvalues are even integers, and the GSO($-$) 
component is made of states whose $L_0+\frac{1}{2}$ eigenvalues are odd 
integers.}. These are exactly the 
same formal expressions which give the pure gauge and tachyon vacuum 
solutions of \cite{Erler}. The only 
new ingredient here is $f$, which can depend on $G$. Explicitly,
\begin{equation}f=f_++Gf_-,\end{equation} where $f_\pm=f_\pm(K)$ are functions 
of $K$ only. In terms of $f_\pm$ the solution takes the form
\begin{equation}\Psi[f] = \left(cKB\frac{1-f_++Gf_-}{(1-f_+)^2-Kf_-^2}c
+B\gamma^2\right)(f_++Gf_-).\label{eq:sol2}\end{equation}
With a little extra work we can also compute $\sqrt{f_++Gf_-}$ to find the 
twist symmetric solution. 

The physical interpretation of these solutions depends on the choice of 
$f_\pm$. To see how, it is helpful 
to employ a formal analysis in the $\mathcal{L}^-$ level expansion, which is 
an easy and apparently reliable method for identifying gauge orbits in 
solutions of this type\cite{simple}. Recall that 
$\mathcal{L}^-=\mathcal{L}_0-\mathcal{L}_0^\star$ 
is a reparameterization generator and a derivation. This means that 
the star product of two $\mathcal{L}^-$ eigenstates is itself an eigenstate, 
and the eigenvalues add. Since $K,B,c,G$ are eigenstates of $\mathcal{L}^-$
(see table \ref{tab:alg}), we can find the $\mathcal{L}^-$ level expansion 
of $\Psi[f]$ by expanding in powers of $K$ and ordering the terms in 
sequence of increasing scaling dimension. The expansion can take one of three
different forms, depending on the behavior of $f_\pm$ at $K=0$:
\begin{eqnarray}\mathrm{Pure\ Gauge}: \lineup\ \ \ \ f_+(0)\neq 1,
\label{eq:bc_gauge}\\
\mathrm{Half\ Brane}: \lineup\ \ \ \ f_+(0)=1,\ \ \ f_-(0) \neq 0,
\label{eq:bc_exotic}\\
\mathrm{Tachyon\ Vacuum}: \lineup\ \ \ \ f_+(0)= 1,\ \ \ f_-(0)=0,\ \ \ 
f_+'(0)\neq 0, \label{eq:bc_vacuum}
\end{eqnarray}
corresponding to the expansions
\begin{eqnarray}\mathrm{Pure\ Gauge}:\ \ \ \lineup 
\Psi = \frac{f_+(0)}{1-f_+(0)}\,Q(Bc)- \frac{f_+(0)^2}{1-f_+(0)}\,B\gamma^2
\ + \ ...,\nonumber\\
\mathrm{Half\ Brane}:\ \ \  \lineup \Psi = -\frac{1}{f_-(0)}\,cGBc\ + 
\ ...,\nonumber\\
\mathrm{Tachyon\ Vacuum}:\ \ \  \lineup \Psi = -\frac{1}{f_+'(0)}\,c\ + 
\ ...,\end{eqnarray}
where $...$ denotes higher level terms. Each expansion formally 
corresponds to a physically distinct gauge orbit\footnote{Following 
\cite{simple}, one can construct a formal gauge transformation
relating solutions with different choices of $f$: 
$\Psi[f'] = g^{-1}(Q+\Psi[f])g$. However, the gauge transformation breaks 
down if $f$ and $f'$ do not share the same boundary conditions at $K=0$, 
\eq{bc_gauge}-\eq{bc_vacuum}, since either $g$ or $g^{-1}$ would formally
require inverse powers of $K$ in its $\mathcal{L}^-$ level expansion. Inverse 
powers of $K$ are not constructible states within the wedge algebra.} 
within our general ansatz (see figure \ref{fig:exotic_3sol}). The pure gauge 
and tachyon vacuum solutions are known\cite{Erler}, but the so-called
{\it half-brane solutions} are new. These are the main subject of this 
paper.

\begin{figure}
\begin{center}
\resizebox{3.5in}{1.9in}{\includegraphics{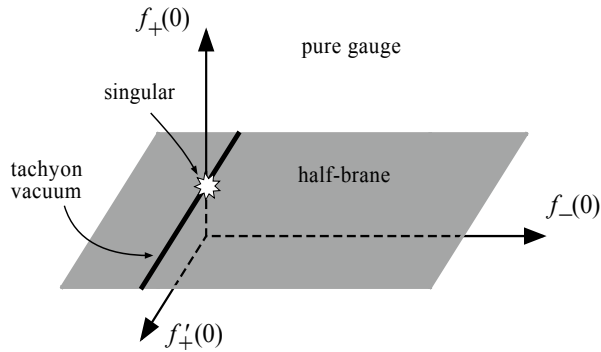}}
\end{center}
\caption{\label{fig:exotic_3sol} Three dimensional ``phase space'' of 
solutions, parameterized by $f_+(0),f_-(0)$ and $f_+'(0)$. Tachyon vacuum 
solutions sit on a line embedded in a plane of half-brane solutions, which 
themselves are embedded in an ambient space of pure gauge solutions. 
Note that the point $f_+(0)=1,f_-(0)=0,f_+'(0)=0$ appears to represent 
singular solutions.} 
\end{figure}

Often one can get insight into the physics of a solution by inspecting its 
leading term in the $\mathcal{L}^-$ level expansion. For the 
tachyon vacuum the leading term is proportional to the $c$ ghost, which
is responsible for the absence of cohomology at the vacuum\cite{Erler}. For 
pure gauge solutions, the leading term is BRST exact to linear order, 
corresponding to the fact that pure gauge solutions represent a deformation 
of the perturbative vacuum by a trivial element of the BRST cohomology. 
For half brane solutions, the full meaning of the leading term $cGBc$ is 
not clear to us. However, it is worth noting that $cGBc$ has twist 
eigenvalue $+i$:
\begin{equation}(cGBc)^\S = +i\, cGBc.\end{equation}
Therefore half-brane solutions result from condensation of
states in the GSO($-$) sector carrying odd integer eigenvalues of 
$L_0+\frac{1}{2}$. This is peculiar since all of these states carry positive 
mass squared. The more familiar states responsible for tachyon 
condensation carry even integer $L_0+\frac{1}{2}$, and in fact these states 
have vanishing expectation value in the twist even solution \eq{sol_real}. 
The fact that half-brane solutions result from ``condensation'' of massive 
modes of the open string is one way to anticipate that they cannot satisfy the
reality condition.

Let us give two explicit examples of half-brane solutions. The first comes 
by setting
\begin{equation}f= f_++Gf_- = \frac{1}{1-iG},\end{equation}
and takes the form
\begin{equation}\Psimp = \Big[i cGBc +Q(Bc)\Big]\frac{1+iG}{1+K}.
\label{eq:simple}\end{equation}
We will explain the factor of $i$ shortly. We will call this the 
{\it simple} half-brane solution, 
since it is in many ways analogous to the ``simple'' tachyon vacuum 
introduced in \cite{simple}. In 
particular, \eq{simple} requires no phantom term, and gives
the most straightforward calculation of the action and gauge invariant 
overlap. Another solution, which is likely to be better behaved in the 
level expansion (see appendix \ref{app:level} and \cite{simple}), comes from 
setting
\begin{equation}f= f_++Gf_- = (1+ia G)\Omega,\label{eq:Sch_f}\end{equation}
where $a\neq 0$ is a parameter. It takes the form,
\begin{equation}
\Psch = \left[c\frac{KB(1-\Omega+ia G\Omega)}{(1-\Omega)^2+a^2K\Omega^2}c
+B\gamma^2\right](1+ia G)\Omega.\label{eq:Schnabl}\end{equation}
Unlike \eq{simple}, this solution is composed of wedge states whose angles 
have strictly positive lower bound. We will call it the 
{\it Schnabl-like} solution. To compute the action 
or gauge invariant overlap, we should express 
the solution as a regularized sum subtracted against a phantom term. We will 
explain how to do this in section \ref{sec:phantom}. 

\subsection{Half-Brane Solutions in Berkovits' String Field Theory}
\label{subsec:berkovits}

We would now like to know whether half-brane solutions exist in Berkovits' 
string field theory. The task is to find a pair of string fields 
$(g,g^{-1})$ at ghost and picture number zero, and in the large Hilbert space,
satisfying
\begin{eqnarray}Qg \lineup = g\Psi, \label{eq:Berk1}\\ g^{-1}g \lineup = 
gg^{-1}= 1.\label{eq:Berk2}\end{eqnarray}
where $\Psi$ is a cubic half-brane solution. Within a certain subalgebra of 
states, we will show that these equations have no solutions for $g$ and 
$g^{-1}$. A similar approach can be used to argue that Berkovits' string 
field theory does not have a tachyon vacuum solution on a BPS D-brane\cite{ES}.

To solve equations \eq{Berk1} and \eq{Berk2}, we must extend our subalgebra
to include fields in the large Hilbert space. The minimal and most natural
extension is to include the string field
\begin{equation}A = -\sigma_3\otimes \xi\d\xi e^{-2\phi}c(1)|I\rangle
\end{equation}
which satisfies
\begin{equation}QA=1,\ \ \ \ A\gamma^2=-c,\ \ \ \ Ac=cA=0,\ \ \ \ 
\llbracket \gamma,A\rrbracket=0,\ \ \ \ \llbracket \d c,A\rrbracket=0.
\end{equation}
$A$ describes an insertion of an inverse picture changing operator 
multiplied by the $\xi$ zero mode. It has ghost number $-1$, is
effective Grassmann odd, carries even worldsheet spinor number, and has 
scaling dimension $0$. Naively, the field $A$ is enough to generate any 
Berkovits solution given any cubic solution. To see how, note that
\begin{equation}g=1+A\Psi.\label{eq:formal_berk}\end{equation}
solves \eq{Berk1}\cite{equivalence,super_photon}. Then, we can {\it almost} 
solve \eq{Berk2} by expressing $g^{-1}$ as an infinite geometric series in 
powers of $-A\Psi$. However, this series is not guaranteed to converge. 
This is why the cubic and Berkovits equations of motion are not 
{\it a priori} equivalent. 

We search for a Berkovits half-brane by making the most general possible ansatz
in the subalgebra generated by $K,B,c,G$ and $A$. Expand $\Psi$ and 
$(g,g^{-1})$ into $\mathcal{L}^-$ eigenstates as follows:
\begin{eqnarray}\Psi\ \lineup = \,
\Psi_{-1/2}+\Psi_0+\Psi_{1/2}+...,\nonumber\\
g\ \lineup = \, g_{-1/2}\, +\,g_0\,+\,g_{1/2}\,+...,\nonumber\\
g^{-1}\lineup = \, \bar{g}_{-1/2}\, +\,\bar{g}_0\,+\,\bar{g}_{1/2}\,+...,
\end{eqnarray}
where the subscripts refers to the $\frac{1}{2}\mathcal{L}^-$ eigenvalue of the
fields. If $\Psi$ is a cubic half-brane solution, its expansion 
takes the general form
\begin{eqnarray}\Psi_{-1/2} \lineup=\, -\frac{1}{\alpha_1+\alpha_2}cGBc,
\nonumber\\
\Psi_0\ \, \lineup= \,
-\frac{\alpha_1}{\alpha_1+\alpha_2}GcGBc-\frac{\alpha_2}{\alpha_1+\alpha_2}
cGBcG+\frac{\beta_1+\beta_2}{(\alpha_1+\alpha_2)^2}cKBc+B\gamma^2,
\nonumber\\
\Psi_{1/2}\ \lineup =\, ...,
\label{eq:lowest_psi}\end{eqnarray}
where $\alpha_1,\alpha_2,\beta_1,\beta_2$ are constants parameterizing the 
gauge orbit up to this level. The most general ansatz for $g$ is 
\begin{eqnarray}g_{-1/2}\,\lineup = x\, A\gamma,\nonumber\\
g_0\ \lineup = y_1+y_2\, Bc+y_3\,A \d c+y_4\,GA\gamma+y_5\,A\gamma G,
\nonumber\\
g_{1/2}\,\lineup = ...,
\label{eq:lowest_g}\end{eqnarray}
where $x$ and $y_1,...,y_5$ are 
coefficients to be determined by solving the equations of motion. We make a 
similar ansatz for $g^{-1}$. Now plug 
these formulas into \eq{Berk1} and solve level by level:
\begin{eqnarray}0\lineup = g_{-1/2}\Psi_{-1/2},\nonumber\\
Qg_{-1/2}\lineup = g_{-1/2}\Psi_0+g_0\Psi_{-1/2},\nonumber\\
\lineup\,\vdots\ \ \ \ \ \ \ \ \ \ \ \ \ \ \ \ \ \ \ \ \ \ \ \ \ \ \ .
\label{eq:Berk1_levels}\end{eqnarray} 
The lowest level equation is trivially satisfied. Plugging \eq{lowest_psi} 
and \eq{lowest_g} into the next equation gives
\begin{equation}xQ(A\gamma) = 
\frac{2i \alpha_1 x-y_1+2iy_5}{\alpha_1+\alpha_2} cGBc+x\, \gamma cB.
\end{equation}
Note that the right hand side is in the small Hilbert space. Acting with 
$\eta_0$ therefore gives
\begin{equation}x\,Q(\eta_0(A\gamma))=0\label{eq:xcoeff}\end{equation}
The field $\eta_0(A\gamma)$ is the zero momentum tachyon in the $-1$ 
picture. Since the zero momentum tachyon is off-shell, \eq{xcoeff} implies
that the coefficient $x$ vanishes, i.e. $g_{-1/2}=0$. 
A similar argument also shows that $\bar{g}_{-1/2}=0$.
Equation \eq{Berk1_levels} then implies that $g_0$ has a right kernel:
\begin{equation}g_0\Psi_{-1/2}=0.\end{equation}
To construct $g^{-1}$, we must solve \eq{Berk2} level by level:
\begin{eqnarray}\bar{g}_{-1/2}g_{-1/2} \lineup = 0,\nonumber\\
\bar{g}_{-1/2}g_0 + \bar{g}_0g_{-1/2} \lineup =0,\nonumber\\
\bar{g}_{-1/2}g_{1/2} +\bar{g}_0g_0+\bar{g}_{1/2}g_{-1/2}\lineup = 1,
\nonumber\\
\lineup\ \vdots \ \ \, .\end{eqnarray}
Since $g_{-1/2}=\bar{g}_{-1/2}=0$ this implies
\begin{equation}\bar{g}_0g_0=1,\end{equation}
but this contradicts the fact that $g_0$ has a right kernel. Therefore, 
the Berkovits half-brane solution does not exist in the $K,B,c,G,A$ 
subalgebra. While it is possible that a more general ansatz is necessary, 
we believe that this subalgebra is rich enough to capture a half-brane 
solution, if one were to exist.\footnote{Fuchs and Kroyter 
suggest\cite{equivalence} a general mapping between cubic and Berkovits 
solutions $g=1+\tilde{A}\Psi$, where $\tilde{A}$ is a midpoint insertion of 
$\xi \d\xi e^{-2\phi}c$. However, this solution appears to be too singular 
to allow for a computation of the Berkovits action.}

\section{Reality Condition}
\label{sec:reality}

Physical solutions in cubic superstring field theory are expected to satisfy 
the reality condition\footnote{The dagger ($^\ddag$) refers to a composition 
of Hermitian and BPZ conjugation. See appendix 
\ref{app:conventions}. Note that this form of the reality condition is correct
only for the left-handed star product convention.}\cite{Ohmori,Zwiebach}
\begin{equation}\Psi^\ddag = \Psi,\label{eq:strong_real}\end{equation}
It is important to ask whether half-brane solutions meet this criterion. 
Surprisingly, the answer is no, according to the following theorem: 
\begin{theorem}
Under assumptions {\bf 1)}-{\bf 4)} stated below, there are no half-brane
solutions in the $K,B,c,G$ subalgebra satisfying the reality condition.
\end{theorem}

\begin{proof} Every solution $\Psi$ in the $K,B,c,G$ subalgebra is 
associated with a pair of states in wedge algebra:
\begin{equation}f_+(K),\ \ \ \ \ f_-(K).\end{equation}
We can reconstruct $f_+$ and $f_-$ from the solution by 
solving the equations\footnote{The most general solution in the 
$K,B,c,G$ subalgebra can be found by making the most general 
(formal) pure gauge ansatz, following Okawa\cite{Okawa}. In equation \eq{fpm}
$\beta$ represents a line integral insertion of the $\beta$ ghost in the 
sliver coordinate frame.}
\begin{equation}B\Psi B = B\frac{K(f_++Gf_-)}{1-(f_++Gf_-)}. \ \ \ \ \ \ \ \ 
\llbracket\beta,\llbracket\beta,\Psi\rrbracket\rrbracket 
= B(f_++Gf_-),\label{eq:fpm}
\end{equation}
To prove the theorem, we show that the reality condition is inconsistent with 
certain regularity conditions which must be imposed on $f_+(K)$ and $f_-(K)$. 
The regularity conditions are:
\begin{description}
\item {\bf 1)} $f_+(0)=1$ and $f_-(0)\neq 0$ and in particular $f_+'(0)$ is 
finite.
\item {\bf 2)} $\lim_{K\to\infty}f_+(K) =0 $ and 
$\lim_{K\to\infty}\sqrt{K}f_-(K)=0$.
\item {\bf 3)} $f_+$ and $f_-$ are continuous functions of $K$ for all 
$K\geq 0$.
\item {\bf 4)} The field $\frac{K}{(1-f_+)^2-Kf_-^2}$, is a continuous 
function of $K$ for all $K\geq 0$.
\end{description}
Condition {\bf 1)} is essentially the definition of the half-brane solution.
Condition {\bf 2)} ensures that the solution is not too ``identity-like,'' so
that it can have well-defined action and gauge invariant overlap. Conditions
{\bf 3)} and {\bf 4)} are motivated by a conjecture due to
Rastelli\cite{Rastelli} suggesting that the algebra of wedge states should be 
identified with the $C^*$-algebra of bounded, continuous functions on the 
positive real line\footnote{The definition of the algebra of wedge states is
not known, but discontinuous functions of $K$ appear to be problematic in
the level expansion. For related discussions, 
see \cite{Schnabl_lightning}.}. In particular, {\bf 3)} and {\bf 4)} 
assume that $f_+,f_-$ and $\frac{K}{(1-f_+)^2-Kf_-^2}$ must 
separately be well-defined states in order for the solution itself 
to be well-defined. Since these fields can be extracted directly from the 
solution via equation \eq{fpm}, this assumption appears necessary.

The reality condition implies that $f_+$ and $f_-$ are real 
functions of $K$. To see why this contradicts regularity, consider the 
denominator of the expression appearing in {\bf 4)}, which we will call $D(K)$:
\begin{equation}D(K) = (1-f_+)^2 - Kf_-^2.\end{equation}
By assumption {\bf 1)} we have
\begin{equation}D(0)=0\end{equation}
and
\begin{equation}D'(0) = -f_-(0)^2. \end{equation}
Since the slope is negative, we have
\begin{equation}D(K)<0\ \  \mathrm{for\ some\ positive}\ K.
\label{eq:pf1}\end{equation}
Now by assumption {\bf 2)}
\begin{equation}\lim_{K\to\infty}D(K) = 1.\label{eq:pf2}\end{equation}
Since $D$ is continuous by assumption {\bf 3)}, this means 
\begin{equation}D(K)=0\ \  \mathrm{for\ some\ strictly\ positive}\ K.
\end{equation}
See figure \ref{fig:exotic_real}. Since $D(K)$ has a zero, the ratio 
$K/D$ cannot be continuous for all $K\geq 0$ which violates assumption 
{\bf 4)}.
\end{proof}

\begin{figure}
\begin{center}
\resizebox{3in}{1.55in}{\includegraphics{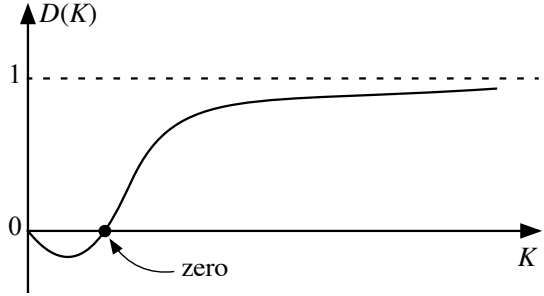}}
\end{center}
\caption{\label{fig:exotic_real}
If $f_\pm(K)$ are real, boundary conditions for the half-brane solution
at $K=0$ and $\infty$ require that the denominator of the solution
\eq{sol2} has a zero for positive $K$.}
\end{figure}

It is helpful to see why real $f_+$ and $f_-$ are problematic in specific 
examples. Suppose we defined the simple solution in \eq{simple} 
without the factor of $i$:
\begin{equation}f=f_++Gf_-=\frac{1}{1-G},\end{equation}
In this case condition {\bf 4)} is satisfied since
\begin{equation}\frac{K}{(1-f_+)^2-Kf_-^2}=K-1\end{equation}
is a continuous function of $K$. But condition {\bf 3)} is not satisfied: 
both $f_+$ and $f_-$ are equal to $\frac{1}{1-K}$, which has a pole at $K=1$. One could try to define 
$\frac{1}{1-K}$ using the Schwinger parameterization\cite{simple} 
\begin{equation}\frac{1}{1-K} = -\int_0^\infty dt\, e^t \Omega^t,\end{equation}
but since the wedge state $\Omega^t$ approaches a constant 
(the sliver) for large $t$, this integral diverges exponentially. A 
second example
is the Schnabl-like solution with $a=-i$, so that the factor of $i$ cancels
in \eq{Sch_f}. In this case
\begin{equation}f_+=f_-=\Omega\end{equation}
are both real and satisfy {\bf 3)}, but
\begin{equation}\frac{K}{(1-f_+)^2-Kf_-^2} = 
\frac{K}{(1-\Omega)^2 - K\Omega^2}\label{eq:ex2}\end{equation}
has a pole at $K\approx 0.931$ and violates {\bf 4)}. One can try to define 
this state by a geometric series
\begin{equation}\frac{K(1-\Omega)}{(1-\Omega)^2 - K\Omega^2} = 
K(1-\Omega)\left[\sum_{n=0}^\infty (2\Omega-(1-K)
\Omega^2)^n \right],\label{eq:geom_ex}\end{equation}
and evaluate the contribution of each term in the series 
to a typical state in the Fock space, for example $L_{-2}|0\rangle$. 
We have found numerically that the contributions to this coefficient
eventually grow exponentially with $n$. By contrast, 
contributions from the analogous sum at $a=1$ decay quite rapidly 
(as $1/n^4$) to give the coefficient $2.86\,L_{-2}|0\rangle$. 

\subsection{Topology of Half-Brane Solutions}
\label{subsec:Topology}

Though half-brane solutions are not real, every half-brane solution is 
related to its conjugate by a complex gauge transformation:
\begin{equation}\Psi^\ddag = U^{-1}(Q+\Psi)U.\label{eq:weak_real}
\end{equation}
The required $U$ is straightforward to compute (see appendix B of 
\cite{simple}) and is regular, in as far as the solutions themselves are 
regular. This raises an interesting issue. From the perspective of gauge 
invariant observables, a solution satisfying \eq{weak_real} is naively 
equivalent to a real solution. In fact, such solutions have been useful
for studying marginal deformations with singular OPEs\cite{FK,KO}, solutions 
in Berkovits' string field 
theory\cite{super_photon,super_marg,Okawa_supermarg1,Okawa_supermarg2,KOsuper},
and the tachyon vacuum \cite{simple}. 

A second thought, however, reveals that \eq{weak_real} is not quite
enough to guarantee the reality of observables. We must also require that $U$ 
can be implemented as a sequence of infinitesimal gauge transformations, 
that is, $U$ can be continuously deformed to the identity. Remarkably, for 
half-brane solutions, this appears not to be possible. That is, $U$ is a 
topologically nontrivial gauge transformation. This 
is the first explicit example of a topologically nontrivial 
gauge transformation in string field theory, and is especially interesting 
since the topology is not related to any spacetime geometry in the 
$\alpha'\to0$ limit, but appears to be intrinsic to the internal structure 
of the string.

To start we must define what it means to ``continuously deform'' the gauge 
transformation $U$. For simplicity, we will restrict ourselves to the $K,B,c,G$
subalgebra, though we presume that our results are more general.
We assume that a continuous deformation of the gauge transformation $U$
will effect a continuous deformation of half brane solutions, 
in the following sense:
\begin{definition}
(Continuity.) Let $\Psi(t),t\in[0,1]$ be a 1-parameter family of half-brane 
solutions in the $K,B,c,G$ subalgebra. We say that this family is 
{\bf continuous} only if $f_+(K,t)$ and $f_-(K,t)$, defined via \eq{fpm}, 
satisfy the following properties
\begin{description}
\item{\bf A1)-A4)} $f_+(K,t)$ and $f_-(K,t)$ satisfy conditions 
{\bf 1)}-{\bf 4)} for every fixed $t\in[0,1]$. 
\item{\bf B)} $f_+(K,t)$ and $f_-(K,t)$ are continuous functions of 
$K$ and $t$ for $K\geq 0$ and $t\in[0,1]$. 
\end{description}
\end{definition}
\noindent Conditions {\bf A1)-A4)} ensure that $\Psi(t)$ is a regular 
half-brane solution for all $t$. Condition {\bf B)} ensures that there are 
no ``jumps'' as we change $t$, that is, $f_+$ and $f_-$ should change 
continuously with $t$ if $\Psi(t)$ does. We now come to our central claim:
\begin{theorem}
There is no continuous 1-parameter family of half-brane solutions 
$\Psi(t),t\in[0,1]$ in the $K,B,c,G$ subalgebra such that 
$\Psi(0)=\Psi(1)^\ddag$. 
\end{theorem}
\noindent This means, in particular, that the gauge transformation $U$ relating
a half-brane solution to its conjugate cannot be continuously deformed to 
the identity. 
\begin{proof} Since $\Psi(0)=\Psi(1)^\ddag$, the family of states 
$f_+(K,t)$ and $f_-(K,t)$ associated with $\Psi(t)$ must satisfy the 
boundary condition, 
\begin{equation}f_+(K,0)=f_+(K,1)^*,\ \ \ \ f_-(K,0)=f_-(K,1)^*.
\label{eq:fpmconj}\end{equation}
Analogous to the proof of Theorem 1, we will show that this boundary 
condition is incompatible with the continuity conditions stated above. In 
particular, we will show that the boundary condition, together with 
conditions {\bf A1)-A3)} and {\bf B)} imply that the field 
\begin{equation}D=(1-f_+)^2-Kf_-^2\end{equation} 
has a zero at some point $(K,t)$. Therefore condition {\bf A4)} is 
violated, and the sought after continuous family of solutions does not exist. 

It is useful to think of $K$ and $t$ as coordinates on a semi-infinite strip
\begin{equation}\Sigma = \mathbb{R}_+\otimes[0,1],\end{equation}
Consider the function
\begin{equation}\left.\Theta\right|_{\delta\Sigma}
=\left.\frac{D}{|D|}\right|_{\delta\Sigma},
\end{equation}
which maps the boundary of $\Sigma$ into complex numbers of unit modulus. 
The boundary includes the point at $K=\infty$, so that $\delta\Sigma$ has 
the topology of a circle. We make the following claims:

\begin{claim}$\Theta|_{\delta\Sigma}$ is a continuous map from $\delta\Sigma$
into complex numbers of unit modulus.
\end{claim}

\begin{proof} By assumption we take the half-brane solution and its conjugate
at $t=0$ and $t=1$ to be well-defined solutions. Therefore, $D$ cannot have 
any zeros for positive $K$ at $t=0$ and $t=1$. Conditions 
{\bf A1)}, {\bf A2)}, {\bf B)} then imply continuity on all of $\delta\Sigma$.
\end{proof}
 
\begin{claim} If $\Theta|_{\delta\Sigma}$ has nonzero winding number, 
then $D$ has a zero inside $\Sigma$.
\end{claim}

\begin{proof} Suppose $D$ has no zeros in $\Sigma$. Since {\bf B)} implies 
that $D$ is continuous, we can extend 
$\Theta|_{\delta\Sigma}$ to a continuous function on the entire semi-infinite
strip by simply taking $\Theta=\frac{D}{|D|}$. Shrinking $\delta\Sigma$ 
to a point, this function gives a continuous homotopy from 
$\Theta|_{\delta\Sigma}$ to the identity map. Since the identity map has zero
winding number, the result follows.
\end{proof}

\begin{claim} The winding number of $\Theta|_{\delta\Sigma}$ is odd.\end{claim}

\begin{proof} The proof of this claim is the most technical part of the 
argument. As a first step, it is helpful to introduce a notion of ``winding 
number'' for maps from a closed interval into complex numbers of 
unit modulus. Let $g$ be a continuous map from a 
closed oriented interval $I$ into complex numbers of unit modulus. 
We can lift $g$ to a continuous map $\phi:I\to\mathbb{R}$ such that 
$g=e^{i\phi}$. Parameterizing $I$ by 
$\lambda\in[0,1]$, we define the {\it winding number} of $g$ to be the 
unique integer $n$ such that
\begin{equation}\phi(1)-\phi(0)=2\pi n +R,\ \ \ \ \ 0\leq R<2\pi. 
\label{eq:open_winding}\end{equation}
We will write $n=w[g]$. 

Now consider two closed oriented intervals $I_1$ and $I_2$ which
intersect at their endpoints to form a circle $S^1$. Assume that the 
orientation of $I_1$ is the same as that of the circle, and the 
orientation of $I_2$ is opposite. If $g$ is a continuous map from $S^1$ 
into complex numbers of unit modulus, then
\begin{equation}w[g] = w\left[g|_{I_1}\right]-w\left[g|_{I_2}\right],
\label{eq:seg_add}\end{equation}
where $g|_{I_1}$ and $g|_{I_2}$ is the restriction of $g$ to the
intervals $I_1$ and $I_2$, respectively. The proof is 
straightforward.

\begin{figure}
\begin{center}
\resizebox{5in}{1.2in}{\includegraphics{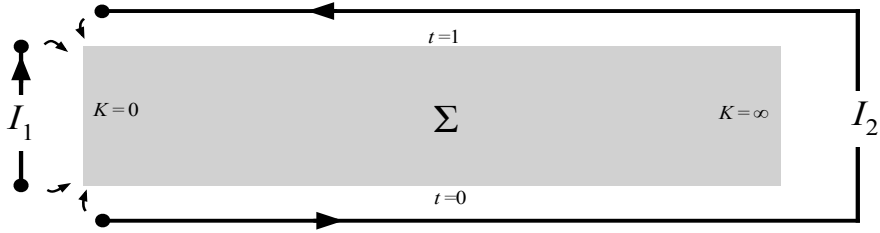}}
\end{center}
\caption{\label{fig:Sigma} The boundary segments $I_1$ and $I_2$ of $\Sigma$.
}
\end{figure}

We compute the winding number of $\Theta|_{\delta\Sigma}$ by
splitting $\delta\Sigma$ into two segments, computing the winding numbers on 
each segment separately, and taking the difference following \eq{seg_add}. 
The segments will be: 
\begin{eqnarray}
\lineup I_1:\ \mathrm{the}\ K=0\ \mathrm{boundary\ of}\ \Sigma,\nonumber\\
\lineup I_2:\ \mathrm{the}\ t=0\ \mathrm{and}\ t=1\ \mathrm{boundaries\ of}\
\Sigma,\ \mathrm{connected\ through}\ K=\infty.\nonumber
\end{eqnarray}
See figure \ref{fig:Sigma}. First we compute the winding number of 
$\Theta|_{I_1}$ in terms of the winding number of the function
\begin{equation}\left.\frac{f_-}{|f_-|}\right|_{I_1}=e^{i\theta},
\label{eq:fmodfm}\end{equation}
Conditions {\bf A1)} and {\bf B)} implies that $\theta$ is a continuous map 
from $I_1$ into $\mathbb{R}$. The winding number of \eq{fmodfm} is the 
integer $n$ satisfying
\begin{equation}\theta(1)-\theta(0)=2\pi n +R\ \ \ \ 0\leq R<2\pi.
\label{eq:FmW}\end{equation}
From {\bf A1)} we also have 
\begin{equation}\Theta|_{I_1}=e^{i(2\theta+\pi)}.
\label{eq:DmodD_I1}\end{equation}
The winding number of $\Theta|_{I_1}$ follows:
\begin{equation}(2\theta(1)+\pi)-(2\theta(0)+\pi)=4\pi n+2R.\end{equation}
If $0\leq R<\pi$, the winding number is $2n$, and if $\pi\leq R<2\pi$
the winding number is $2n+1$. Thus
\begin{eqnarray}w[\Theta|_{I_1}]\lineup 
\in 2\mathbb{Z}\ \ \ \ \ \ \ \ \ \ \ \mathrm{if}\ \ 0\leq R<\pi,\nonumber\\
w[\Theta|_{I_1}]\lineup \in 2\mathbb{Z}+1
\ \ \ \ \ \ \mathrm{if}\ \ \pi\leq R<2\pi.
\end{eqnarray}
Now compute the winding number on $I_2$. Parameterize $I_2$
by $\lambda\in[0,1]$ such that: 1) $\lambda=0$ corresponds to the 
corner of the strip where $K$ and $t$ both vanish; 2) 
$\lambda=\frac{1}{2}$ corresponds the the point at $K=\infty$, and 
3) $\lambda=1$ corresponds to the corner where $K$ vanishes and $t=1$. 
Also write
\begin{equation}\Theta|_{I_2}=e^{i\psi},
\end{equation}
where $\psi$ is a continuous map from $I_2$ into $\mathbb{R}$.
Because of \eq{fpmconj}, $\Theta|_{I_2}$ evaluated on the $t=1$ boundary 
is the complex conjugate of its value on the $t=0$ boundary, which implies 
\begin{equation}\psi(1)-\psi(0) = 2(\psi(1)-\psi(\half)).\end{equation}
Continuity requires that $\Theta=1$ at $K=\infty$, and therefore 
$\psi(\frac{1}{2})=2\pi m$ for some integer $m$:
\begin{equation}\psi(1)-\psi(0) = 2\psi(1)-4\pi m.\end{equation}
Comparing with equation \eq{DmodD_I1} we are free to assume 
\begin{equation}\psi(1)=2\theta(1)+\pi\end{equation} 
Also \eq{fpmconj} implies
\begin{equation}\theta(1)=-\theta(0)+2\pi k,\end{equation}
for some integer $k$. Therefore
\begin{eqnarray}\psi(1)-\psi(0) \lineup = 4\theta(1)+2\pi(-2m+1)\nonumber\\
\lineup = 2(\theta(1)-\theta(0))+2\pi(2k-2m+1).
\end{eqnarray}
Now we substitute equation \eq{FmW} we find
\begin{equation}\psi(1)-\psi(0) = 2\pi(2n+2k-2m+1)+2R.\end{equation}
Therefore
\begin{eqnarray}w[\Theta|_{I_2}]\lineup 
\in 2\mathbb{Z}+1\ \ \ \ \ \ \ \mathrm{if}\ \ 0\leq R<\pi,\nonumber\\
w[\Theta|_{I_2}]\lineup \in 2\mathbb{Z}
\ \ \ \ \ \ \ \ \ \ \ \ \mathrm{if}\ \ \pi\leq R<2\pi.
\end{eqnarray}
Now we invoke \eq{seg_add} to find the final result:
\begin{equation}w[\Theta|_{\delta \Sigma}]
=w[\Theta|_{I_1}]-w[\Theta|_{I_2}]\in 2\mathbb{Z}+1.\end{equation}
\end{proof}
\noindent Since the winding number of $\Theta|_{\delta \Sigma}$ is odd, it 
cannot be zero. Therefore $D$ must vanish at some point inside the 
semi-infinite strip. This completes the proof of Theorem 2.
\end{proof}

Theorem 2 implies that half-brane solutions come in at least two 
topologically distinct sectors in the $K,B,c,G$ subalgebra. With some
extra work, one can show that there are precisely two. It would be nice
to find a way to characterize these sectors in terms of an invariant which is
computable in terms of $f_+$ and $f_-$. A related question is that, though
we have claimed that every half-brane solution is physically distinguishable 
from its conjugate on account of a topologically nontrivial gauge 
transformation, we have not found an observable which would actually 
distinguish between a half-brane solution and its conjugate in practice. 
The energy and closed string overlap, computed in section 
\ref{sec:observables}, are real observables for all 
half-brane solutions. Finding an observable which can detect the failure of 
the reality condition would give much insight into the physical significance of
half-brane solutions, as well as their topological structure.

\section{Regularization and Phantom Piece}
\label{sec:phantom}

In this section we discuss the regularization and phantom term for the 
half-brane solution. For clarity we focus on the Schnabl-like half-brane 
solution, though the discussion can be extended to solutions based on more
general choices of $f_\pm$.

The Schnabl-like solution \eq{Schnabl} can be written in the form
\begin{equation}
\Psch = \left[c\frac{KB}{1-(1+iaG)\Omega}c
+B\gamma^2\right](1+ia G)\Omega.\label{eq:Schnabl2}\end{equation}
It is useful\cite{Erler} to replace the factor between the $c$ insertions 
by the partial sum of a geometric series, with the appropriate ``error term'':
\begin{equation}\frac{K}{1-(1+iaG)\Omega} = \sum_{n=0}^N K[(1+iaG)\Omega]^n
+\frac{K}{1-(1+iaG)\Omega}[(1+iaG)\Omega]^{N+1}.\end{equation}
With this substitution, the solution can be written in the form 
\begin{equation}\Psch = \Psi_{N+1}-\sum_{n=0}^N\psi_n'+\Gamma,
\label{eq:sum_phantom}\end{equation}
where, reflecting the notation of Schnabl\cite{Schnabl}, we have defined
the fields 
\begin{eqnarray}
\psi_n' \lineup = -cKB\Big[(1+iaG)\Omega\Big]^n c\Big[(1+iaG)\Omega
\Big]\\
\Gamma \lineup = B\gamma^2\Big[(1+iaG)\Omega\Big] \\
\Psi_{N+1}\lineup 
= c\left(\frac{KB}{1-(1+iaG)\Omega}\Big[(1+iaG)\Omega\Big]^{N+1}
\right)c\Big[(1+iaG)\Omega\Big]\label{eq:phantom}
\end{eqnarray}
Since we have just made a trivial substitution, \eq{sum_phantom} is equal
to the Schnabl-like solution for all $N$. However this expression is most
useful in the $N\to\infty$ limit. In this limit $\Psi_{N+1}$ becomes the 
so-called ``phantom term,'' and vanishes in the Fock space. 
Compared with phantom terms for the tachyon vacuum 
solution\cite{Erler,Schnabl,SSF2}, the large $N$ limit of $\Psi_{N+1}$ is 
novel and requires careful treatment.

To understand the phantom term we should study the large $N$ behavior of 
the string field 
\begin{equation} \Big[(1+iaG)\Omega\Big]^N. \end{equation}
It is helpful to decompose this into GSO($\pm$) components as follows: 
\begin{equation}\Big[(1+iaG)\Omega\Big]^N= X_N + (ia N G) Y_N,\label{eq:XY}
\end{equation}
where
\begin{eqnarray}
X_N \lineup = 
\frac{1}{2}\Big[(1+ia\sqrt{K})^N+(1-ia\sqrt{K})^N\Big]\Omega^N,\nonumber\\
Y_N\lineup = \frac{1}{2iaN\sqrt{K}}
\Big[(1+ia\sqrt{K})^N-(1-ia\sqrt{K})^N\Big]
\Omega^N .\label{eq:XYnosum}\end{eqnarray}
Above we introduced $\sqrt{K}$ formally in order to write closed form 
expressions for the sums:
\begin{eqnarray}
X_N \lineup = 
\sum_{0\leq k\leq N/2}{N \choose 2k}a^{2k}(-K)^k
\Omega^N,\nonumber\\
Y_N\lineup = \frac{1}{N}\sum_{0\leq k\leq \frac{N-1}{2}}{N \choose 2k+1}
a^{2k}(-K)^k \Omega^N. \label{eq:ph_sums}\end{eqnarray}
A naive argument suggests that the large $N$ limit of $X_N$ and $Y_N$ should be 
divergent. Note that while $(-K)^k\Omega^N$ vanishes as a power 
$\frac{1}{N^{2+k}}$ in the Fock space, the binomial coefficients grow very
rapidly, so for fixed $k$ and large $N$ a term in the sum for $X_N$ diverges 
as a power
\begin{equation}{N\choose 2k}(-K)^k\Omega^N\sim N^{k-2},
\end{equation}
Generically a sum of such terms would diverge faster than any 
power of $N$ in the Fock space. 

Miraculously, however, for a certain range of the parameter $a$ 
$X_N$ and $Y_N$ converge to the sliver state: 
\begin{equation}\lim_{N\to\infty}X_N =\lim_{N\to\infty}Y_N 
=\Omega^\infty.\label{eq:fNpm}
\end{equation}
To see how this happens, consider the Fock space expansion for the wedge 
state $\Omega^\alpha$. We can write the expansion in the form
\begin{equation}|\Omega^\alpha\rangle = 
\sum_{\vec{n}}P_{\vec{n}}\left(\frac{1}{\alpha+1}\right)\,L_{-n_q}...
L_{-n_2}L_{-n_1}|0\rangle,
\label{eq:fock_wedge}\end{equation}
where $\vec{n}=(n_q,...n_2,n_1)$ is a list of integers of arbitrary length
satisfying
\begin{equation}n_q\geq... \geq n_2\geq n_1\geq 2,\end{equation}
and $P_{\vec{n}}(x)$ are a collection of polynomials in $x$ which determine the
coefficients of $|0\rangle$ and its descendents. For example, up to level 
$4$ the nonvanishing polynomials are
\begin{eqnarray}P_{()}(x) \lineup = 1,
\ \ \ \ \ \ \ \ \ \ \ \ \ \ \ \ \ \ \ \ \ \ \ \ \ \,
P_{(2)}(x) = -\frac{1}{3} + \frac{4}{3}x^2,\nonumber\\
P_{(2,2)}(x) \lineup = \frac{1}{9} - \frac{8}{9}x^2+\frac{16}{9}x^4,
\ \ \ \ \ \ 
P_{(4)}(x) = \frac{1}{30}-\frac{16}{30}x^4.
\end{eqnarray}
To compute $X_N,Y_N$, replace the factors of $K$ multiplying $\Omega^N$ in
\eq{ph_sums} with derivatives via the formula,
\begin{equation}(-K)^k\Omega^N = 
\left.\frac{d^k}{d\alpha^k}\Omega^\alpha\right|_{\alpha=N},\end{equation}
and plug in the Fock space expansion \eq{fock_wedge}. The coefficient of the 
state labeled by $\vec{n}$ will then be
\begin{eqnarray}\sum_{0\leq k\leq N/2}{N \choose 2k}a^{2k}
\left.\frac{d^k}{d\alpha^k}
P_{\vec{n}}\left(\frac{1}{\alpha+1}\right)\right|_{\alpha=N}\lineup
\ \ \ \ \mathrm{for}\ \ 
X_N,\nonumber\\
\frac{1}{N} \sum_{0\leq k\leq \frac{N-1}{2}}{N \choose 2k+1}
a^{2k}\left.\frac{d^k}{d\alpha^k}P_{\vec{n}}\left(\frac{1}{\alpha+1}\right)
\right|_{\alpha=N}\lineup\ \ \ \ \mathrm{for}\ \ \, Y_N.
\label{eq:fpm_coef}\end{eqnarray}
The miracle of convergence as $N\to\infty$ is due to the following 
identities: 
\begin{eqnarray}\lim_{N\to\infty}\left[
\sum_{0\leq k\leq N/2}{N \choose 2k}a^{2k}
\left.\frac{d^k}{d\alpha^k}\frac{1}{(\alpha+1)^h}\right|_{\alpha=N}\right]
\lineup = 0, \label{eq:limits1}\\
\lim_{N\to\infty}\left[\sum_{0\leq k\leq \frac{N-1}{2}}{N \choose 2k+1}
a^{2k}\left.\frac{d^k}{d\alpha^k}\frac{1}{(\alpha+1)^h}\right|_{\alpha=N}
\right]\lineup =0,\label{eq:limits2}\end{eqnarray}
which assume 
\begin{equation}h>0,\ \ \ a\in[-\sqrt{2},\sqrt{2}].\end{equation}
Thus taking the $N\to\infty$ of equation \eq{fpm_coef}, all nonzero powers 
of $\frac{1}{1+\alpha}$ are killed, leaving
\begin{eqnarray}
P_{\vec{n}}(0)\lineup\ \ \ \ \mathrm{for}\ \ 
\lim_{N\to\infty}X_N,\nonumber\\
P_{\vec{n}}(0)\lineup\ \ \ \ \mathrm{for}\ \ 
\lim_{N\to\infty}Y_N.
\end{eqnarray}
These are exactly the coefficients of the sliver state.

To prove the identities \eq{limits1} and \eq{limits2}, it is helpful to 
represent the ratio $\frac{1}{1+\alpha}$ as an integral: 
\begin{equation}\frac{1}{(\alpha+1)^h} = \frac{1}{(h-1)!}\int_0^\infty dt\, 
t^{h-1} e^{-(\alpha+1)t}.\end{equation}
Substituting in \eq{limits1} converts the sum into a simple integral:
\begin{eqnarray}\lineup \sum_{0\leq k\leq N/2}{N \choose 2k}a^{2k}
\left.\frac{d^k}{d\alpha^k}\frac{1}{(\alpha+1)^h}\right|_{\alpha=N}\nonumber\\
\lineup\ \ \ \ \ \ \ \ \ 
= \frac{1}{(h-1)!}\int_0^\infty dt\, t^{h-1}e^{-t}\frac{1}{2}\left(
\Big[(1+ia\sqrt{t})e^{-t}\Big]^N+\Big[(1-ia\sqrt{t})e^{-t}\Big]^N\right).
\label{eq:int_limit}\end{eqnarray}
Note the similarity of the integrand with \eq{XYnosum}. Let us assume that the two terms in the integrand should be bounded in 
absolute value in the $N\to\infty$ limit. This requires
\begin{equation}|(1\pm ia\sqrt{t})e^{-t}|^2\leq 1,\end{equation}
which implies that $a$ must be a real number in the range 
$[\sqrt{2},-\sqrt{2}]$.\footnote{Numerical computations 
suggest that that \eq{limits1} may hold even in a limited region of the 
complex plane around the line segment $-\sqrt{2}\leq a\leq\sqrt{2}$. We do not
have an analytic proof, however.} Now we can compute the $N\to\infty$
limit in the integrand by simply noting
\begin{equation}\lim_{N\to\infty}\Big[(1\pm ia\sqrt{t})e^{-t}\Big]^N = 
\left\{\begin{matrix}1 &\ \ \mathrm{at}\ t=0\\
0 & \ \ \mathrm{otherwise}\end{matrix}\right..
\end{equation}
Since the integrand vanishes almost everywhere, the sum \eq{limits1} must 
vanish. This completes the proof that $X_N$ and $Y_N$ approach this sliver
state in the large $N$ limit. 

\begin{figure}
\begin{center}
\resizebox{3.9in}{2.7in}{\includegraphics{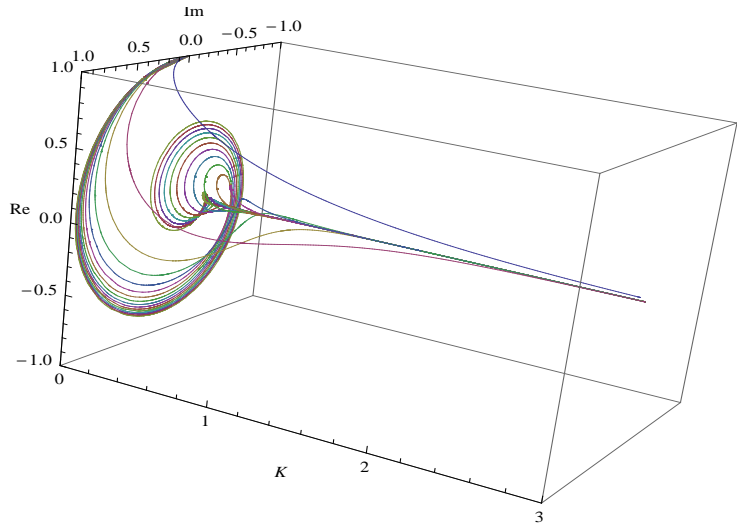}}
\end{center}
\caption{\label{fig:exotic_phantom} Plots of $[(1+ i\sqrt{K})e^{-K}]^N$ as a 
function of $K$ for $N=1,3,5,...31$. The 
vertical axis is the real part, the forward axis is the imaginary part, and
the horizontal axis is $K$.}
\end{figure}

Perhaps this result is not surprising, since the sliver state is the 
only projector which could have emerged from the $N\to\infty$ limit.  
What is more interesting is the manner in which $X_N$ and $Y_N$ approach the 
sliver. Recall that the sliver state corresponds to a function of $K$
which takes the value $1$ at $K=0$ and vanishes everywhere else. 
$X_N$ and $Y_N$ approach the sliver by a sequence of functions 
\begin{equation}[(1+ ia\sqrt{K})e^{-K}]^N.\label{eq:XYmod}\end{equation}
A plot of the real and imaginary parts of these functions is shown in figure 
\ref{fig:exotic_phantom}. The plot reveals a spiral, which as $N$ 
increases, winds increasingly many times around the $K$ axis and 
becomes increasingly damped away from $K=0$. As $N\to\infty$, the function 
$[(1+ ia\sqrt{K})e^{-K}]^N$ vanishes for $K>0$, but for infinitesimal $K$ 
winds so many times that it fills the whole unit disk in the complex plane. 
Now compare this to how wedge states approach the sliver at large wedge 
angle, corresponding to the sequence of functions $e^{-NK}$. The large $N$ 
limit of $e^{-NK}$ reveals none of the highly oscillatory behavior seen in 
figure \ref{fig:exotic_phantom}. We can formalize this qualitative 
observation as follows. Assume, following a proposal of 
Rastelli\cite{Rastelli}, that convergence in the wedge algebra is determined 
by the norm
\begin{equation}||f(K)|| = \sup|f(K)|.\label{eq:Rast_norm}\end{equation}
In the large $N$ limit, $X_N$ and $Y_N$ could be considered ``close'' to 
a wedge state if there were a set of numbers $\sigma(N),\rho(N)$ such that the 
sequence of norms
\begin{equation}||X_N - \Omega^{\sigma(N)}||,\ \ \ \ 
||Y_N - \Omega^{\rho(N)}||\label{eq:bad_seq}\end{equation}
converged to zero. However, this is impossible; $\Omega^{\sigma(N)}$ 
only takes values between $0$ and $1$, whereas $X_N$ and $Y_N$ take all 
values in the interval $[-1,1]$ for sufficiently large $N$. We can make a 
similar observation by looking at the Fock space expansion. In appendix 
\ref{app:largeN} we compute the leading order corrections to the sums 
\eq{limits1} and \eq{limits2}:
\begin{eqnarray}\sum_{0\leq k\leq N/2}{N \choose 2k}a^{2k}
\left.\frac{d^k}{d\alpha^k}\frac{1}{(\alpha+1)^h}\right|_{\alpha=N}
\lineup = (-1)^h \frac{2(2h-1)!}{(h-1)!}\frac{1}{(aN)^{2h}} + ... ,\nonumber\\ 
\sum_{0\leq k\leq \frac{N-1}{2}}{N \choose 2k+1}
a^{2k}\left.\frac{d^k}{d\alpha^k}\frac{1}{(\alpha+1)^h}\right|_{\alpha=N}
\lineup =(-1)^{h+1}
\frac{4N(2h-3)!}{(h-2)!}\frac{1}{(aN)^{2h}}+....\label{eq:reg_limits}
\end{eqnarray}
Plugging these into \eq{fpm_coef}, we find that the leading order
correction to $X_N$ and $Y_N$ for large $N$ is numerically quite different 
from that of a wedge state with large wedge angle, especially due to the 
$h$-dependent factors in \eq{reg_limits}. In principle these differences 
could have an important effect on calculations involving the phantom term. 
Therefore, though $X_N$ and $Y_N$ approach the sliver, it is not 
the ``same'' sliver as $\Omega^N$ for large 
$N$.\footnote{Note that the analog of $X_N,Y_N$ for the
tachyon vacuum solutions of \cite{SSF2} is the infinite power $f(K)^N$, 
where $f$ is some function of $K$ satisfying the constraints 
$f(0)=1,f'(0)=-\gamma<0$ and $|f(K)|\leq 1$. We can show that
the leading correction to the Fock space coefficients for any such state 
match those of a wedge state $\Omega^{\gamma N}$. Moreover, the sequence of
norms $||f^N - \Omega^{\gamma N}||$ converges to zero. In this sense, 
$f(K)^N$ is approximately equal to the wedge state $\Omega^{\gamma N}$ for 
large $N$.}

This motivates us to introduce a new class of states, more general than 
wedge states, which could describe the large $N$ behavior of $X_N$ and $Y_N$:
\begin{equation}e^{-\alpha K}e^{i\beta G}= 
\Omega^\alpha\left[\cos(\beta\sqrt{K})+iG\, 
\frac{\sin(\beta\sqrt{K})}{\sqrt{K}}\right].\label{eq:ext_wedge}\end{equation}
We will call these {\it super-wedge states}. It turns out that super-wedge 
states cannot be easily described as a superposition of wedge surfaces. We will
say more about how these states can be constructed in the next section 
and appendix \ref{app:largeN}. Consider the large $N$ limit of the 
super-wedge state
\begin{equation}\Big(e^{iaG}\Omega^{1-\frac{a^2}{2}}
\Big)^N=\hat{X}_N+(iaNG)\hat{Y}_N.
\label{eq:Gwedge}\end{equation}
One can check that the sequence of norms
\begin{equation}||X_N - \hat{X}_N||,\ \ \ \ \ 
||Y_N - \hat{Y}_N||\end{equation}
converges to zero, and moreover the leading large $N$ correction to the 
Fock space coefficients of $(\hat{X}_N,\hat{Y}_N)$ match those of 
$(X_N,Y_N)$. We can therefore simplify the phantom piece for large $N$
by replacing
\begin{equation}\Big[(1+iaG)\Omega\Big]^N\ \
\rightarrow\ \  \left(e^{ia G}\Omega^{1-\frac{a^2}{2}}\right)^N,
\end{equation}
and 
\begin{equation} \frac{K}{1-(1+iaG)\Omega}\ \ \rightarrow\ \ \frac{i}{a}G,
\end{equation}
which is the leading term in the $\mathcal{L}^-$ level 
expansion of this factor.\footnote{Sometimes it is 
necessary to include more than the leading term in the $\mathcal{L}^-$ level 
expansion\cite{Erler}, but ignore this possibility for simplicity.} The phantom
term therefore simplifies to
\begin{equation}\hat{\Psi}_N = 
\frac{i}{a}c\left[ GB \left(e^{iaG}\Omega^{1-\frac{a^2}{2}}\right)^N \right]
c\,(1+iaG)\Omega,
\label{eq:reg_phantom}\end{equation}
and we can express the Schnabl-like solution in regularized form,
\begin{equation}\Psch = \lim_{N\to\infty}\left[\hat{\Psi}_N
-\sum_{n=0}^N\psi_n'+\Gamma\right].\label{eq:reg_sum_phantom}
\end{equation}
Unlike \eq{sum_phantom}, this expression is only valid in the $N\to\infty$ 
limit. 

A few comments about the parameter $a$. Note that the phantom 
term \eq{reg_phantom} is manifestly singular as $a$ approaches zero. This 
corresponds to the fact that at $a=0$ $\Psch$ is actually a solution for the 
the tachyon vacuum, which requires a different phantom piece. 
Also note we needed to fix $a$ to lie within a restricted range 
$-\sqrt{2}\leq a\leq\sqrt{2}$. We do not know whether this bound 
reflects a limitation of the regularization \eq{reg_sum_phantom} or a 
deeper problem with the solution when $a$ sits outside the restricted range.

\subsection{An Aside: General States in the Wedge Algebra}
\label{subsec:gen_f}

Up to now, most states in the wedge algebra have been assumed to 
take the form
\begin{equation}f(K) = \int_0^\infty dt \tilde{f}(t)\Omega^t,
\label{eq:laplace_f}\end{equation}
and thus are a ``linear combination'' of wedge states. When this 
integral converges, the function $f(K)$ can be identified with the Laplace 
transform of $\tilde{f}(t)$. However, the set of functions which can be 
represented as a Laplace transform in this sense is limited. Here we 
would like to give a more general construction of $f(K)$ motivated by our
analysis of the phantom term.
 
Consider the space of polynomials over a variable $x$, which we denote 
$\mathbb{R}[x]$. Here, $x$ will be identified with the ratio 
$\frac{1}{1+\alpha}$ in the Fock space expansion of the wedge state 
$\Omega^\alpha$. Given a suitable function $f(t)$, we define a linear 
functional $L_f$ on $\mathbb{R}[x]$ as follows:
\begin{eqnarray}L_f(x^0) \lineup = f(0),\nonumber\\ 
L_f(x^1) \lineup = \int_0^\infty dt\, f(t)e^{-t},\nonumber\\
\lineup\vdots\nonumber\\ 
L_f(x^h)\lineup = \frac{1}{(h-1)!}\int_0^\infty dt\, t^{h-1}f(t)e^{-t}.
\label{eq:Lf}\end{eqnarray}
With the help of this functional, we define the state,
\begin{equation}f(K) = L_f\left(\Omega^\alpha\right),\ \ \ \ \ 
x=\frac{1}{1+\alpha}.
\label{eq:gen_f}\end{equation}
This should be understood as a definition of $f(K)$ in the Fock 
space. One can check that \eq{gen_f} and \eq{laplace_f} give exactly the same 
expressions for the coefficients of $f(K)$ in the domain 
where both formulas are defined. However \eq{gen_f} is much more general. 
For example, \eq{gen_f} allows one to construct a string field for any $f(K)$
in the algebra of bounded, continuous functions on the positive real line.
The existence of such string fields is implied by Rastelli's proposal for the 
definition of the algebra of wedge states\cite{Rastelli}.

The simplest example of a state which we can construct from \eq{gen_f}, 
but not \eq{laplace_f}, is a wedge state with  complex wedge angle
\begin{equation}\Omega^{\alpha+i\beta}.\end{equation}
The linear functional \eq{Lf} is
\begin{equation}L_{\Omega^{\alpha+i\beta}}(x^h) = 
\frac{1}{(h-1)!}\int_0^\infty dt\, t^{h-1}e^{-(\alpha+1)t+i\beta t}
=\frac{1}{(1+\alpha+i\beta)^h},\ \ \ \ \ \alpha>-1.\label{eq:complex_wedge}
\end{equation}
This is the obvious analytic continuation of the usual Fock space expansion
of a wedge state to complex wedge angle. But note that the linear functional
only converges if $\alpha>-1$. This ``explains'' the curious relation between
the Fock space coefficients of wedge states with positive and negative wedge 
angle:
\begin{equation}\Omega^\alpha=\Omega^{-2-\alpha}\ \ \ ?\end{equation}
In reality an inverse wedge state for $\alpha<-1$ is divergent in the Fock 
space, and not the analytic continuation of \eq{complex_wedge}.

Another example is the state
\begin{equation}f(K) = \frac{\lambda K\Omega}{1-\lambda\Omega}.\label{eq:prgg}
\end{equation}
This is the pure gauge solution of Schnabl\cite{Schnabl},
either in the ghost number zero toy model or in the ghost number one case
after ignoring the $B,c$ insertions. For $|\lambda|<1$ we can express this
as a Laplace transform by making a geometric series 
expansion of the denominator, but for $|\lambda|>1$ this series is divergent. 
Still we can define the linear functional
\begin{equation}L_f(x^h) = \frac{1}{(h-1)!}\int_0^\infty dt 
\frac{\lambda t^h e^{-2t}}{1-\lambda e^{-t}} = \lambda h\,\Phi(\lambda,h,2), 
\ \ \ \lambda\ngtr 1,\end{equation}
where $\Phi(z,s,v)$ is the Lerch function. The integral is absolutely 
convergent as long as $\lambda$ is not a real number greater than one. 
This suggests that the pure gauge solutions are defined even for 
$\lambda<-1$, though the geometric series is divergent. This would bring 
Schnabl's pure gauge solutions into line with the pure gauge solutions of 
\cite{simple}, which are also nonsingular for $\lambda<-1$.

As a final (somewhat peculiar) example, consider the characteristic 
function on the interval $[a,b]>0$:
\begin{equation}{\bf 1}_{[a,b]}(K)= \left\{\begin{matrix}1 & \ \mathrm{for}\ 
a\leq K\leq b\\ 0 & \ \mathrm{otherwise}\end{matrix}\right..\end{equation}
Formally these are all projectors in the wedge algebra, and if the interval
does not include $0$, the projectors are orthogonal to the sliver state.
There is no hope of representing such states as a Laplace transform 
\eq{laplace_f}, but we can still define them using the functional
\begin{equation}L_{{\bf 1}_{[a,b]}}(x^h) = \sum_{n=0}^{h-1}\frac{a^n}{n!}e^{-a}
-\sum_{n=0}^{h-1}\frac{b^n}{n!}e^{-b}.
\end{equation}
Note that these projectors are infinite rank. Unlike the sliver, they are 
difficult to reach by taking the infinite power of a ``reasonable'' $f(K)$ 
in the wedge algebra. While one can apparently define such an $f(K)$ 
using \eq{gen_f}, it would have to satisfy the awkward constraint of being 
exactly equal to unity on the interval $[a,b]$, and strictly less than unity 
in absolute value outside that interval. 

We have not confirmed whether any of these states behave as expected 
under star multiplication. Since they are not linear combinations of surfaces, 
one cannot study their star products using the usual gluing rules of 
conformal field theory. Nevertheless, these states are concrete 
constructions which could be interesting for future study.

\section{Observables}
\label{sec:observables}

\subsection{Gauge Invariant Overlap for Simple Half-Brane Solution}

We start with the simplest computation, that of the gauge invariant
overlap for the simple half-brane solution \eq{simple}:
\begin{equation}W(\Psimp,\mathcal{V}) = 
\llangle\Psimp\rrangle_\mathcal{V}.\label{eq:simp_ov}
\end{equation}
Here the bracket $\llangle\cdot\rrangle_\mathcal{V}$ is defined
in the same way as the vertex $\llangle\cdot\rrangle$ (see 
\eq{left_vertex}) except the picture changing operator $Y_{-2}$ is replaced 
by an on shell closed string vertex operator $\mathcal{V}(i)$ inserted at the 
midpoint. We assume $\mathcal{V}$ is an NS-NS closed string vertex 
operator of the form,
\begin{equation}\mathcal{V}(z) = c\tilde{c}e^{-\phi}e^{-\tilde{\phi}}
\mathcal{O}^{(\frac{1}{2},\frac{1}{2})}(z,\bar{z}),\label{eq:V}\end{equation}
where $\mathcal{O}^{(\frac{1}{2},\frac{1}{2})}$ is a weight $(\half,\half)$ 
superconformal matter primary. We work in the small Hilbert space, so the 
$\xi$ zero mode is absent. If the interpretation of Ellwood\cite{Ellwood} is 
correct, the gauge invariant overlap should represent the shift in the closed 
string tadpole of the solution relative to the perturbative vacuum.

Plugging in the simple solution \eq{simple} into the overlap, the BRST 
exact term does not contribute since $\mathcal{V}$ is on-shell. Furthermore
the GSO($-$) component vanishes in the correlator. This leaves
\begin{equation}W(\Psimp,\mathcal{V}) = 
-\leftllangle cGBc\frac{G}{1+K}\rightrrangle_\mathcal{V}.
\end{equation}
Now expand $\frac{1}{1+K}$ in terms of wedge states
\begin{equation}-\leftllangle cGBc\frac{G}{1+K}\rightrrangle_\mathcal{V} =
-\int_0^\infty dt\, e^{-t}\Bigllangle cGBcG \Omega^t\Bigrrangle_\mathcal{V}.
\end{equation}
Note
\begin{equation}cGBcG\Omega^t = t^{\frac{1}{2}\mathcal{L}^-}(cGBcG\Omega). 
\end{equation}
Since the operator $t^{\frac{1}{2}\mathcal{L}^-}$ is a reparameterization, 
it leaves the bracket invariant. We can then easily evaluate the integral 
over $t$ to find
\begin{equation}W(\Psimp,\mathcal{V})
=-\llangle cGBcG\,\Omega\rrangle_\mathcal{V}.\end{equation}
Now commute the leftmost $G$ insertion towards the other $G$ insertion: 
\begin{eqnarray}W(\Psimp,\mathcal{V}) 
\lineup= -\llangle cB(cG +\delta c)G\,\Omega\rrangle_\mathcal{V}\nonumber\\
\lineup =  
-\llangle cK\,\Omega\rrangle_\mathcal{V}+
2i\llangle c\gamma BG\,\Omega \rrangle_\mathcal{V}. \label{eq:1step_ov}
\end{eqnarray}
To compute the first term note
\begin{equation}-cK\Omega = \half\mathcal{L}^- (c\Omega) +c\Omega .
\label{eq:trick1}\end{equation}
Since $\mathcal{L}^-$ kills the bracket, this leaves 
\begin{equation}W(\Psimp,\mathcal{V})=\llangle c\Omega\rrangle_\mathcal{V}
+2i\llangle c\gamma BG\,\Omega \rrangle_\mathcal{V}.\label{eq:1step_ovf}
\end{equation}
Now focus on the second term. Using cyclicity we can rewrite it in the form, 
\begin{equation}
2i\llangle c\gamma BG\Omega\rrangle_\mathcal{V} = 
i \Bigllangle\, G (c\gamma B\Omega) + (c\gamma B\Omega)G\, 
\Bigrrangle_\mathcal{V}. \label{eq:2step_ov}\end{equation}
Since $\gamma$ carries odd worldsheet spinor number, the $G$
anticommutator above is exactly the worldsheet supersymmetry 
variation $\delta$:
\begin{equation}\delta(c\gamma B\Omega) = G (c\gamma B\Omega)
+ (c\gamma B\Omega) G  .
\end{equation}
Therefore we can explicitly eliminate $G$:
\begin{eqnarray}2i\llangle c\gamma BG\,\Omega\rrangle_\mathcal{V}
\lineup =i\llangle \delta(c\gamma B\Omega)\rrangle_\mathcal{V}\nonumber\\
\lineup =\Bigllangle -2\gamma^2 B\Omega+ \frac{1}{2}c\partial cB\Omega
\Bigrrangle_\mathcal{V},\end{eqnarray}
where we used the derivation property of $\delta$ and the explicit variations
\eq{dother}. To eliminate the remaining $B$ insertion we use the 
derivation $\mathcal{B}^-$, which satisfies 
\begin{equation}\half\mathcal{B}^-K = B,\ \ \ \ \ \ 
\half\mathcal{B}^-\{B,\, c,\,\, \mathrm{or}\, \gamma\} =0,
\end{equation}
and leaves the vertex invariant. This allows us to express
\begin{equation}
-2\gamma^2 B\Omega+ \frac{1}{2}c\partial cB\Omega = 
\mathcal{B^-}\left[-2\gamma^2
\Omega -\frac{1}{2}cKc\Omega\right]-\frac{1}{2}c\Omega.
\end{equation}
The $\mathcal{B}^-$ term kills the bracket leaving
\begin{equation}2i\llangle c\gamma B G\,\Omega\rrangle_\mathcal{V}
 = -\frac{1}{2}\llangle c\Omega\rrangle_\mathcal{V}.
\label{eq:2step_ovf}\end{equation}
Plugging into \eq{1step_ovf} and gives the final answer:
\begin{equation}W(\Psimp,\mathcal{V}) = 
\frac{1}{2}\llangle c\Omega\rrangle_\mathcal{V}.
\label{eq:ov}\end{equation}
This is exactly one-half the value of the overlap at the tachyon
vacuum\cite{Ellwood}. This confirms that the half-brane 
solutions are not gauge equivalent to either the tachyon vacuum or the 
perturbative vacuum (where the overlap vanishes identically). It also 
indicates that half-brane solutions must source closed strings with half
the strength of a non-BPS D-brane.   

\subsection{Gauge Invariant Overlap for Schnabl-Like Half-Brane Solution}

As a check on the consistency of our results, we would like to compute 
the gauge invariant overlap for the Schnabl-like solution. 
Plugging in \eq{sum_phantom} we find
\begin{eqnarray}W(\Psch,\mathcal{V})\lineup = 
\llangle\Psch\rrangle_\mathcal{V},\nonumber\\
\lineup = \llangle\Psi_{N+1}\rrangle_\mathcal{V} - 
\sum_{n=0}^N\llangle \psi_n'\rrangle_\mathcal{V} + \llangle 
\Gamma\rrangle_\mathcal{V}.
\label{eq:sch_ov}\end{eqnarray}
The second two terms do not contribute to the overlap. The
easiest way to see this is to note that the overlap for the 
pure gauge solution
\begin{equation}\Psi_\lambda = -\sum_{n=0}^\infty \lambda^n \psi_n' 
+\lambda\Gamma\label{eq:sch_gauge}\end{equation}
must vanish order by order in $\lambda$ by gauge invariance. 
Therefore \eq{sch_ov} simplifies to
\begin{equation}W(\Psch,\mathcal{V}) = 
\llangle \Psi_N\rrangle. \end{equation}
This equation holds for any $N$, but we will be interested in the limit
$N\to\infty$. 

Plugging in the \eq{phantom} for $\Psi_{N}$ the overlap reduces to 
a sum of four terms:
\begin{eqnarray}W(\Psch,\mathcal{V})
\lineup = \ \ \ 
\leftllangle c\,\frac{KB(1-\Omega)}{(1-\Omega)^2+a^2K\Omega^2}
\,X_N\,c\,\Omega\rightrrangle_\mathcal{V}\nonumber\\ 
\lineup \ \ - Na^2 \leftllangle c\, 
\frac{K^2B\Omega}{(1-\Omega)^2+a^2K\Omega^2}\,
Y_N\,c\,\Omega\rightrrangle_\mathcal{V}\nonumber\\
\lineup \ \ \ -a^2 \leftllangle c\,\frac{KGB\Omega}{(1-\Omega)^2+a^2K\Omega^2}
\,X_N\,c\,G\,\Omega\rightrrangle_\mathcal{V}\nonumber\\
\lineup -Na^2 \leftllangle c\,\frac{KGB(1-\Omega)}{(1-\Omega)^2+a^2K\Omega^2}
\,Y_N\,c\,G\,\Omega\rightrrangle_\mathcal{V}.
\label{eq:4terms}\end{eqnarray}
Now in each of the four terms expand all of the wedge state factors besides 
$X_N$ and $Y_N$ in powers of $K$: 
\begin{eqnarray}
\leftllangle c\,\frac{KB(1-\Omega)}{(1-\Omega)^2+a^2K\Omega^2}
\,X_N\,c\,\Omega\rightrrangle_\mathcal{V}
\lineup = \sum_{m\geq1,n\geq1} C^{(1)}_{mn}\llangle c\,B\,K^m\, X_N\,c\,K^n
\rrangle_\mathcal{V},
\nonumber\\\label{eq:ds1}\\
-Na^2 \leftllangle c\, \frac{K^2B\Omega}{(1-\Omega)^2+a^2K\Omega^2}\,
Y_N\,c\,\Omega\rightrrangle_\mathcal{V} 
\lineup = \sum_{m\geq1,n\geq1}C^{(2)}_{mn}\llangle c\,B\,K^m\, Y_N\,c\,K^n
\rrangle_\mathcal{V},
\nonumber\\\label{eq:ds2}\\
-a^2 \leftllangle c\,\frac{KGB\Omega}{(1-\Omega)^2+a^2K\Omega^2}
\,X_N\,c\,G\,\Omega\rightrrangle_\mathcal{V} 
\lineup= \sum_{m\geq0,n\geq0}
C^{(3)}_{mn}\llangle c\,B\,G\,K^m\, X_N\,c\,G\,K^n\rrangle_\mathcal{V},
\nonumber\\ \label{eq:ds3}\\
-Na^2 \leftllangle c\,\frac{KGB(1-\Omega)}{(1-\Omega)^2+a^2K\Omega^2}
\,Y_N\,c\,G\,\Omega\rightrrangle_\mathcal{V} 
\lineup= \sum_{m\geq1,n\geq0}
C^{(4)}_{mn}\llangle c\,B\,G\,K^m\, Y_N\,c\,G\,K^n
\rrangle_\mathcal{V},
\nonumber\\ \label{eq:ds4}
\end{eqnarray}
where $C_{m,n}^{(a)}$ are constants. Let us focus on \eq{ds3}. To calculate 
the double sum, we should compute the traces 
\begin{equation}
\llangle c\,B\,G\,K^m\, X_N\,c\,G\,K^n\rrangle_\mathcal{V}.
\end{equation}
Plugging in \eq{ph_sums} for $X_N$ this becomes
\begin{equation}\llangle cBGK^m\, X_N\,cGK^n\rrangle_\mathcal{V}  = 
\sum_{0\leq K\leq N/2}{N \choose 2k}a^{2k}\left.\frac{d^k}{d\alpha^k}
\llangle cBG K^m\Omega^\alpha cG K^n\rrangle_\mathcal{V}\right|_{\alpha=N}.
\end{equation}
Now reparameterize the bracket with $\mathcal{L}^-$ to factor the $\alpha$
dependence:  
\begin{equation} \llangle cBGK^m\, X_N\,cGK^n\rrangle_\mathcal{V} = 
\llangle cBG K^m\Omega cG K^n\rrangle_\mathcal{V}
\sum_{0\leq K\leq N/2}{N \choose 2k}a^{2k}\left.\frac{d^k}{d\alpha^k}
\frac{1}{\alpha^{m+n}}
\right|_{\alpha=N}.\end{equation}
We recognize the sum on the right hand side from the identity \eq{limits1}. 
Provided $a\in[-\sqrt{2},\sqrt{2}]$ and $m+n>0$, this vanishes in the 
$N\to \infty$ limit. Therefore if we take $N\to\infty$ only the $m=n=0$ 
term in \eq{ds3} contributes to the overlap:
\begin{equation}
-a^2 \lim_{N\to\infty}
\leftllangle c\,\frac{KGB\Omega}{(1-\Omega)^2+a^2K\Omega^2}
\,X_N\,c\,G\,\Omega\rightrrangle_\mathcal{V} 
= -\llangle cBG \Omega cG\rrangle_\mathcal{V}.\label{eq:sole_cont}
\end{equation}
Now repeat this argument for equations \eq{ds1}, \eq{ds2} and \eq{ds4}.
However, this time the range of summation over $m,n$ excludes all traces 
which could make a nonzero contribution in the $N\to\infty$ limit. So, 
in fact, \eq{sole_cont} is the only contribution to the overlap in the 
$N\to\infty$ limit, and we find
\begin{equation}W(\Psch,\mathcal{V})=
-\llangle cBG \Omega cG\rrangle_\mathcal{V}.\end{equation}
With a few manipulations this can be rewritten,
\begin{equation}W(\Psch,\mathcal{V})=  
-2i\llangle c\gamma BG\Omega\rrangle_\mathcal{V}.\end{equation}
From here on the derivation follows the steps given in 
\eq{2step_ov}-\eq{2step_ovf}, with an extra minus sign, to yield
\begin{equation}W(\Psch,\mathcal{V})=
\frac{1}{2}\llangle c \Omega\rrangle_\mathcal{V}.\end{equation}
in agreement with the overlap for the simple half-brane solution, \eq{ov}.

\subsection{Energy}

Let us now calculate the energy. The energy can be computed from the  
on-shell action:
\begin{equation}E = -S[\Psi] = -\frac{1}{6}\llangle \Psi Q\Psi\rrangle.
\label{eq:energy}\end{equation}
We have only attempted this calculation for the simple half-brane 
solution \eq{simple}, though the final answer should be the same 
for any sufficiently regular half-brane solution. Plugging in $\Psimp$
we find the expression
\begin{equation}E = 
-\frac{1}{6}\left[-\leftllangle cGBc\frac{1}{1+K}Q(cGBc)\frac{1}{1+K}
\rightrrangle
+\leftllangle cGBc\frac{G}{1+K}Q(cGBc)\frac{G}{1+K}\rightrrangle\right].
\end{equation}
To simplify we eliminate the $G$ insertions by repeated 
use of the identity
\begin{equation}\llangle G\Phi \rrangle =\half\llangle \delta\Phi\rrangle.
\label{eq:Gdelta}\end{equation}
The calculation is straightforward, but tedious; the repeated supersymmetry
variations generate dozens of terms. We give some details in appendix 
\ref{app:energy}. In the end, all of the inner products can be evaluated
with the correlation function
\begin{equation}\left\langle Y_{-2}\,c(x_1)c(x_2)\gamma(y_1)\gamma(y_2)
\int_{-i\infty}^{i\infty}\frac{dz}{2\pi i}b(z)\right\rangle_{C_L}
=-\frac{L}{2\pi^2}(x_1-x_2)\cos\frac{\pi(y_1-y_2)}{L}
\label{eq:corr}\end{equation}
evaluated on a cylinder of circumference $L$. Adding 
the terms up, the energy turns out to be
\begin{equation}E=-\frac{1}{4\pi^2},\end{equation}
which is precisely $-1/2$ times the tension of the D-brane. Remarkably, this
is consistent with the computation of the overlap.

\section{Concluding Remarks}

In this paper we have presented a new class of nonperturbative analytic 
solutions of cubic superstring field theory on a non-BPS D-brane. The
nature of these solutions is fundamentally mysterious; they violate the reality
condition, and appear not to exist in Berkovits string field theory. 
Probably they are only formal artifacts of the cubic equations of motion. 
However, their existence is nontrivial and seems significant. We hope that 
further study will shed light into what these solutions represent and why
they exist. 

One immediate consequence of our analysis is that the cubic and Berkovits 
equations of motion are not equivalent. The fact that half-brane solutions
appear to come in two topologically distinct varieties, and the existence of a
 ``tachyon vacuum'' on a BPS D-brane, suggest that the failure of this 
equivalence is related to topological charge. A microscopic understanding 
of D-brane charges is one of 
the longstanding goals of string field theory.  We hope that continued 
development along these lines will give further insight.

\bigskip
\bigskip
\noindent {\bf Acknowledgments}
\bigskip

\noindent The author thanks M. Schnabl and C. Maccaferri for interesting 
conversations, and M. Schnabl and Y. Okawa for critical reading of the 
manuscript. The author also thanks C. Maccaferri 
hospitality during a visit to ULB Brussells, and D. Gross for hospitality 
at the Kavli Institute for Theoretical Physics. This research was supported 
by the EURYI grant GACR EYI/07/E010 from EUROHORC and ESF.

\begin{appendix}

\section{Vertices, Reality and Twist conjugation}
\label{app:conventions}

In this appendix we discuss some important signs connected with the vertices 
and the reality condition in the GSO($-$) sector. Our discussion extends the 
classic analysis of Ohmori\cite{Ohmori} to our preferred ``left handed'' 
star product convention\cite{simple}.

\bigskip
\noindent{\bf Definition of vertices:} In the left handed convention, 
the $N$-string vertex of open string
fields $\Phi_k = \sigma_{i_k}\phi_k(0)|0\rangle$ is defined as a correlator
on the upper half plane as follows:
\begin{equation}\llangle \Phi_1,\Phi_2,...,\Phi_N\rrangle =
\frac{1}{2}\mathrm{tr}(\sigma_3\sigma_{i_1}\sigma_{i_2}...\sigma_{i_N})
\Big\langle Y_{-2}\ f_{1,N}\circ\phi_1(0)\,f_{2,N}\circ\phi_2(0)\,...
\,f_{N,N}\circ\phi_N(0)\Big\rangle,\label{eq:left_vertex}\end{equation}
where 
\begin{equation}Y_{-2} = Y(i)\tilde{Y}(i),\ \ \ 
Y(z) = -\partial\xi e^{-2\phi}c(z), \end{equation}
and the conformal maps defining the vertex are,
\begin{equation}f_{k,N}(z) = -\cot\left[\frac{2}{N}\tan^{-1}z -\frac{\pi}{N}
\left(k-\frac{1}{2}\right)\right].\end{equation}
Let us define
\begin{equation}
x_{k,N} = \cot \frac{\pi}{N}\left(k-\frac{1}{2}\right),\ \ \ \ \ 
y_{k,N}=-\sqrt{\frac{2}{N}}\csc\frac{\pi}{N}\left(k-\frac{1}{2}\right).
\end{equation}
Then $f_{k,N}$ acts on a primary of weight $h$ explicitly as
\begin{equation}f_{k,N}\circ\phi(0) 
= \left(y_{k,N}\right)^{(2h)} \phi(x_{k,N}),
\end{equation}
where the parentheses around $2h$ implies that $h$ must be multiplied by two
{\it before} the power of $y_{k,N}$ is taken. Note that 
\begin{equation}x_{1,N} > x_{2,N} > ... >x_{N,N},\end{equation}
so the position of the vertex operator on the real axis {\it decreases} as we
increase the string label $k$ in the vertex. This is the hallmark of the left 
handed star product convention.

By contrast, in the right handed convention the $N$-string 
vertex would be defined as
\begin{equation}
\llangle \Phi_1,\Phi_2,...,\Phi_N\rrangle_R =
\frac{1}{2}\mathrm{tr}(\sigma_3\sigma_{i_1}\sigma_{i_2}...\sigma_{i_N})
\Big\langle Y_{-2}\ \tilde{f}_{1,N}\circ\phi_1(0)\,
\tilde{f}_{2,N}\circ\phi_2(0)\,...
\,\tilde{f}_{N,N}\circ\phi_N(0)\Big\rangle,\label{eq:right_vertex}
\end{equation}
where the superscript $R$ reminds us that the vertex is defined in the right
handed convention. The conformal maps $\tilde{f}_{k,N}$ are related to 
$f_{k,N}$ simply as,
\begin{equation}\tilde{f}_{k,N}(z) = f_{N+1-k,N}(z).\end{equation}
The positions of the vertex operators on the real axis are
\begin{equation}\tilde{x}_{k,N} = x_{N+1-k,N},\end{equation}
so the position {\it increases} as we increase the string label. This is 
the hallmark of the right handed star product convention. 

The vertices \eq{left_vertex} and \eq{right_vertex} implicitly define the 
open string star product. The left handed star product $\Psi\Phi$ and right 
handed star products $[\Psi\Phi]_R$ are related by the equation,
\begin{equation}[\Psi\Phi]_R = (-1)^{E(\Psi)E(\Phi)+F(\Psi)F(\Phi)}\Phi\Psi.
\end{equation}
The sign appears from anticommuting vertex operators and internal CP factors. 

Let us consider the 2-string vertex. Following \cite{Ohmori}, we define
the action of the BPZ conformal map $I(z) = -\frac{1}{z}$ on a primary of 
weight $h$
\begin{equation}I\circ\phi(z) = \phi(z)^\star 
= \frac{1}{z^{(2h)}}\phi\left(-\frac{1}{z}\right).\end{equation}
By an $SL(2,\mathbb{R})$ transformation, one can then show that the 
2-string vertex in the left and right handed conventions is given by,
\begin{eqnarray}\llangle \Phi_1,\Phi_2\rrangle \lineup =
\frac{1}{2}\mathrm{tr}(\sigma_3\sigma_{i_1}\sigma_{i_2})
\Big\langle Y_{-2}\ \phi_1(0)\,\phi_2(0)^\star\Big\rangle,\nonumber\\
\llangle \Phi_1,\Phi_2\rrangle_R \lineup =
\frac{1}{2}\mathrm{tr}(\sigma_3\sigma_{i_1}\sigma_{i_2})
\Big\langle Y_{-2}\ \phi_1(0)^\star\,\phi_2(0)\Big\rangle.\end{eqnarray}
Now transform the left handed vertex with $I(z)$ and note
\begin{equation}^{\star\star}=(-1)^F.\end{equation}
Then
\begin{equation}\llangle \Phi_1,\Phi_2\rrangle=
(-1)^F\llangle \Phi_1,\Phi_2\rrangle_R.
\end{equation}
So the 2-vertex differs by a sign in the GSO($-$) sector between the two 
star product conventions. This sign plays an important role in fixing the 
string field reality condition. 

\bigskip
\noindent{\bf Reality conjugation:} To formulate the string field reality 
condition, we need to define reality conjugation. To do this it is 
helpful to express the string field in the operator formalism. 
We can define Hermitian and BPZ conjugation of a state $|\Psi\rangle$ or dual 
state $\langle \Psi|$ using the following rules:
\begin{eqnarray}\Big[\,\mathcal{O}|0\rangle\,\Big]^\dag \lineup 
= \langle 0|\mathcal{O}^\dag,\ \ \ \ \ \ 
\Big[\,\langle 0|\mathcal{O}\,\Big]^\dag =\mathcal{O}^\dag|0\rangle, 
\nonumber\\
\Big[\,\mathcal{O}|0\rangle\,\Big]^\star \lineup 
= \langle 0|\mathcal{O}^\star,\ \ \ \ \ \ 
\Big[\,\langle 0|\mathcal{O}\,\Big]^\star =\mathcal{O}^\star|0\rangle,
\end{eqnarray}
and\footnote{We use the six pointed
star $^*$ to denote complex conjugation and the five pointed star $^\star$ to
denote BPZ conjugation, hopefully without confusion.} 
\begin{eqnarray}
(\mathcal{O}_1\mathcal{O}_2)^\dag \lineup = 
\mathcal{O}_2^\dag\mathcal{O}_1^\dag,\ \ \ \ \ \ \ \ \ \ \ \ \ \ \ \ \ \ \ 
\ \ \,
(\mathcal{O}_1\mathcal{O}_2)^\star = 
(-1)^{\mathcal{O}_1\mathcal{O}_2}\mathcal{O}_2^\star\mathcal{O}_1^\star,
\nonumber\\
(a\mathcal{O}_1 + b\mathcal{O}_2)^\dag \lineup = a^*\mathcal{O}_1^\dag + 
b^*\mathcal{O}_2^\dag\nonumber,
\ \ \ \ \ \ (a\mathcal{O}_1 + b\mathcal{O}_2)^\star = a\mathcal{O}_1^\star + 
b\mathcal{O}_2^\star,\ \ \ \ \ \ \ \ a,b\in\mathbb{C}
\end{eqnarray}
and 
\begin{equation}\phi(z)^\dag = \frac{(-1)^\nu}{\bar{z}^{(2h)}}
\phi\left(\frac{1}{\bar{z}}\right),\ \ \ \ \phi(z)^\star = \frac{1}{z^{(2h)}}
\phi\left(-\frac{1}{z}\right).\end{equation}
Here $\mathcal{O}$ are general CFT operators, $a,b$ are complex constants 
and $\phi(z)$ is a primary of weight $h$. The sign $(-1)^\nu$ is needed to 
distinguish between Hermitian and antihermitian fields. We will take the 
$\beta$ ghost to be antihermitian, so that the $\gamma$ ghost 
and the worldsheet supercurrent are Hermitian. Also, we define BPZ conjugation
to leave internal CP factors invariant, whereas Hermitian conjugation takes 
their conjugate transpose. With these definitions we define 
{\it reality conjugation} of a string field $\Psi$ as
\begin{equation}\Psi^\ddag = \Psi^{\dag\star}.
\label{eq:real_conj}\end{equation}
Note that the order matters: first we perform Hermitian and {\it then} BPZ
conjugation to compute the reality conjugate. In particular,
\begin{equation}^{\dag\star} =\, ^{\star\dag}(-1)^F.\end{equation}
Since $^{\star\star}=(-1)^F$ this implies that reality conjugation is 
idempotent,
\begin{equation}^{\ddag\ddag} = 1,\end{equation}
and therefore analogous to complex conjugation. Reality conjugation 
satisfies the important properties
\begin{equation}
(\Psi\Phi)^\ddag = \Phi^\ddag\Psi^\ddag,\ \ \ \ \ \ \ 
(a\Psi+b\Phi)^\ddag = a^*\Psi^\ddag+b^*\Phi^\ddag,\ \ \ a,b\in\mathbb{C},
\label{eq:real_conj_prop1}\end{equation}
and 
\begin{equation}\llangle \Psi\rrangle^* = 
\llangle\Psi^\ddag\rrangle.\label{eq:real_conj_prop2}\end{equation}
Equations \eq{real_conj}-\eq{real_conj_prop2} hold with the appropriate
CP factors attached to the string field. 

\bigskip

\noindent {\bf Twist conjugation:}
We define {\it twist conjugation}
\begin{equation}\Psi^\S = e^{i\pi N}\Psi, \end{equation}
where $N$ is the number operator (the zero momentum component of $L_0$).
For the bosonic string, twist conjugation is related to the 
twist operator $\Omega$ of \cite{Zwiebach,Review} by a 
sign:
\begin{equation}\Psi^\S = -\Omega\Psi.\end{equation}
Twist conjugation satisfies
\begin{eqnarray}(\Psi\Phi)^\S \lineup 
= (-1)^{E(\Psi)E(\Phi)+F(\Psi)F(\Phi)}\Phi^\S
\Psi^\S,\nonumber\\
(a\Psi+b\Phi)^\S \lineup 
= a\Psi^\S+b\Phi^\S\ \ \ \ a,b\in \mathbb{C},\nonumber\\
^{\S\S}\lineup = (-1)^F.\end{eqnarray}
and 
\begin{eqnarray}\llangle\Psi^\S\rrangle = \llangle\Psi
\rrangle.\end{eqnarray}
Twist conjugation gives a way to map between string fields in the left and
right-handed star product conventions\cite{simple}. Suppose $\Phi$ is a 
string field in a theory with left handed star product. The equivalent
field in the right handed theory is 
\begin{equation}\Phi'=\Phi^\S.\label{eq:lr_map}\end{equation}
In the right handed convention, the string field reality 
condition is\cite{Ohmori}:
\begin{equation}(\Psi')^\ddag=(-1)^F\Psi'.\end{equation}
Using \eq{lr_map} it follows that the reality condition for the left handed 
theory is
\begin{equation}\Psi^\ddag=\Psi.\end{equation}
Note that this means that real string fields in the two conventions differ by
a factor of $i$ in the GSO($-$) sector. This factor of $i$ corrects the sign
discrepancy between the 2-vertices, so the tachyon field has the correct sign
kinetic term in either convention.

\section{Superconformal Generator in the Sliver Frame}
\label{app:G}

The string field $G$ is closely related to the superconformal generator 
$G_{-1/2}$ in the sliver conformal frame:
\begin{equation}\mathcal{G} = f_S^{-1}\circ G_{-1/2} = 
\oint \frac{d\xi}{2\pi i}\left(\sqrt{\frac{\pi}{2}}\sqrt{1+\xi^2}\right)
G(\xi).\end{equation}
Since this operator is crucial to the construction of the half-brane 
solution, it is worth understanding in more detail. 

Consider an operator of the form
\begin{equation}{\bm \phi}[f]=\oint \frac{d\xi}{2\pi i}f(\xi)\phi(\xi),
\end{equation}
where $\phi$ is a primary of weight $h$, $f(\xi)$ is a function,
and the contour passes inside an annulus of analyticity of $f(\xi)$ around 
the unit circle. We define Hermitian, BPZ, and dual conjugation\cite{RZ} 
of this operator, respectively,
\begin{eqnarray}
{\bm \phi}[f]^\dag \lineup = {\bm \phi}[f^\dag],\ \ \ \ \ \ \ \ \ \ \ 
f^\dag(\xi) = (-1)^\nu\xi^{(2h-2)}f^*\left(\xi^{-1}\right),\nonumber\\
{\bm \phi}[f]^\star \lineup = {\bm \phi}[f^\star],
 \ \ \ \ \ \ \ \ \ \ \, 
f^\star(\xi) = -(-\xi)^{(2h-2)}f\left(-\xi^{-1}\right),\nonumber\\
{\widetilde{\bm \phi}}[f]\ \lineup= {\bm \phi}[\tilde{f}],
 \ \ \ \ \ \ \ \ \ \ \ \ {\tilde f}(\xi)\ = \epsilon(\xi)f(\xi).
\end{eqnarray}
Here $\epsilon(\xi)$ is the step function\footnote{$\epsilon(\xi)$ has a 
branch cut extending across the 
entire imaginary axis. To define dual conjugation carefully, one should 
represent $\epsilon(\xi)$ as the limit of a sequence of functions which are 
analytic in some (vanishingly thin) annulus containing the unit 
circle\cite{RZ}.}
\begin{equation}\epsilon(\xi) = \left\{ { 1\ \mathrm{for}\ \mathrm{Re}(\xi)>0
\atop -1\ \mathrm{for}\ \mathrm{Re}(\xi)<0}\right..\end{equation}
We also define the combinations
\begin{eqnarray}
{\bm \phi}[f]^+ \lineup = {\bm \phi}[f]+{\bm \phi}[f]^\star,
\ \ \ \ \,\ \ \ \ \ 
{\bm \phi}[f]^-={\bm \phi}[f]-{\bm \phi}[f]^\star,\nonumber\\
{\bm \phi}[f]_L \lineup = \frac{1}{2}\Big({\bm \phi}[f]
+{\widetilde {\bm \phi}}[f]\Big),\ \ \ \ \ 
{\bm \phi}[f]_R=\frac{1}{2}\Big({\bm \phi}[f]-{\widetilde {\bm \phi}}[f]
\Big).\end{eqnarray}
The subscripts $L$ and $R$ denote the left and right halves of the 
charge ${\bm \phi}[f]$. In some cases, the action of ${\bm \phi}[f]_L$ and 
${\bm \phi}[f]_R$ on a state can be described by left or right 
star multiplication with the appropriate string field. When this is possible, 
we say that ${\bm \phi}[f]$ has a {\it non-anomalous} left/right decomposition.

Consider the operators 
\begin{eqnarray}\mathcal{L} \lineup  
= {\bf T}[\ell],\ \ \ \ \ \ \ 
\ell(\xi)=(1+\xi^2)\tan^{-1}\xi ,\nonumber\\
\mathcal{G} \lineup = {\bf G}[g],
\ \ \ \ \ \ \ g(\xi) = \sqrt{\frac{\pi}{2}}\sqrt{1+\xi^2}.
\end{eqnarray}
where 
\begin{equation}{\bf T}[v]=\oint \frac{d\xi}{2\pi i}v(\xi)T(\xi),\ \ \ \ \ \ 
{\bf G}[s] = \oint \frac{d\xi}{2\pi i} s(\xi)G(\xi).\end{equation}
The first is the familiar $\mathcal{L}_0$ of Schnabl\cite{Schnabl}, and the 
second is the operator $G_{-1/2}$ in the sliver conformal frame. The 
functions $\ell(\xi)$ and $g(\xi)$ have branch points at $+i$ 
and $-i$, connected by a branch cut on the imaginary axis passing through 
infinity. The branch points of $\ell(\xi)$ takes the form $x\ln x$ for small 
$x=\xi\pm i$, whereas those of $g(\xi)$ take the form $\sqrt{x}$. The BPZ 
conjugate operators are 
\begin{eqnarray}
\mathcal{L}^\star \lineup = {\bf T}[\ell^\star],\ \ \ \ \ \ \ell^\star(\xi)=
(1+\xi^2)\tan^{-1}\frac{1}{\xi},\nonumber\\
\mathcal{G}^\star \lineup = {\bf G}[g^\star],\ \ \ \ \ \ g^\star(\xi)^=
\sqrt{\frac{\pi}{2}}\xi\sqrt{1+\frac{1}{\xi^2}}.\label{eq:LsGs}
\end{eqnarray}
$\ell^\star,g^\star$ also have branch points at $\pm i$, but the cuts now 
extend on the imaginary axis through the origin. Note that by factoring
$\xi$ into the square root in \eq{LsGs}, $g^\star$ formally appears to be the 
same as $g$. In fact, they are equal up to a sign:
\begin{equation}g^\star(\xi)=\epsilon(\xi)g(\xi).\end{equation}
This means that the BPZ conjugate of $\mathcal{G}$ is equal to
its dual conjugate:
\begin{equation}\mathcal{G}^\star=\widetilde{\mathcal{G}}.
\label{eq:BPZ_dual}
\end{equation}
It is also useful to consider the operators
\begin{eqnarray}
\mathcal{L}^+ \lineup = \mathcal{L}+\mathcal{L}^\star={\bf T}[\ell^+],
\ \ \ \ \ \ \ell^+(\xi)=
\frac{\pi}{2}(1+\xi^2)\epsilon(\xi),\nonumber\\
\tilde{\mathcal{L}}^+ \lineup = {\bf T}[\tilde{\ell}^+],\ \ \ \ \ \ \ \ \ \ \ 
\ \ \ \ \ \ \ \ \tilde{\ell}^+(\xi)=\frac{\pi}{2}(1+\xi^2).
\end{eqnarray}
The Hermitian conjugates of $\mathcal{L},\mathcal{L}^*$ and $\mathcal{G}$
are equal to their BPZ conjugates. For $\mathcal{G}^\star$ and 
$\tilde{\mathcal{L}}^+$ there is a sign difference: 
$\mathcal{G}^{\star\dag} = \mathcal{G} = 
-\mathcal{G}^{\star\star}$ and $\tilde{\mathcal{L}}^{+\dag}
=\tilde{\mathcal{L}}^+=-\tilde{\mathcal{L}}^{+\star}$.

The string fields $K$ and $G$ can be defined through the action of
$\mathcal{L}^+,\mathcal{G}$, and their dual conjugates on a test state:
\begin{eqnarray}\mathcal{L}^+\Phi \lineup = K\Phi+\Phi K,
\ \ \ \ \ \ \ \ \ \ \ \ \ \ \ \ \ \ \ \ 
\widetilde{\mathcal{L}}^+\Phi = K\Phi-\Phi K = \partial\Phi,\nonumber\\
\sigma_1\mathcal{G}^\star\Phi\lineup = G\Phi + (-1)^{F(\Phi)}\Phi G,\ \ \ \,
\ \ \ \ \ 
\sigma_1\mathcal{G}\Phi
= G\Phi-(-1)^{F(\Phi)}\Phi G = \delta\Phi.\label{eq:KGop}
\end{eqnarray}
This definition implies three consistency conditions: 
\begin{eqnarray}
\lineup {\bf 1)}\ \  \tilde{\mathcal{L}}^+|I\rangle = 0
\ \ \ \ \mathcal{G}|I\rangle=0\nonumber\\ 
\lineup {\bf 2)}\ \  \mathcal{L}^+_L (\Phi\Psi) = (\mathcal{L}^+_L\Phi)\Psi
\ \ \ \ \mathcal{L}^+_R (\Phi\Psi) = \Phi(\mathcal{L}^+_R\Psi)\nonumber\\
\lineup \ \ \ \ \ \mathcal{G}_L (\Phi\Psi) = (\mathcal{G}_L\Phi)\Psi
\ \ \ \ \ \, 
\mathcal{G}_R (\Phi\Psi) = (-1)^{\epsilon(\Phi)}\Phi(\mathcal{G}_R\Psi)
\nonumber\\
\lineup {\bf 3)}\ \ \{\mathcal{G}_L,\mathcal{G}_R\}=0,\ \ \ \ \ \  
[\mathcal{G}_L,\mathcal{L}^+_R]=0,\ \ \ \ \ \ 
[\mathcal{L}^+_L,\mathcal{L}^+_R]=0
\end{eqnarray}
The first condition follows from setting $\Phi=|I\rangle$ in \eq{KGop}. 
The second follows from associativity of the star product. The third 
condition follows from the assumption that $\mathcal{L}^+,\mathcal{G}$,
and their dual conjugates should have well defined action on $K$ and $G$;
in other words, $K$ and $G$ can be consistently star multiplied 
among themselves. If these three conditions are satisfied, $\mathcal{L}^+$ and
$\mathcal{G}$ have a non-anomalous left/right decomposition.

It is not difficult to verify conditions {\bf 1)} and {\bf 2)} by contracting
with ``reasonable'' test states (for example, wedge states of positive width
with insertions placed away from the midpoint) and mapping to the upper half 
plane. Condition {\bf 3)} is more subtle and is worth checking explicitly. 
We can compute the commutators using the superconformal algebra expressed 
in the form
\begin{eqnarray}\Big[{\bf T}[v_1],{\bf T}[v_2]\Big]\, \lineup =\, 
{\bf T}[v_2\partial v_1-v_1\partial v_2],\nonumber\\ 
\Big[{\bf T}[v],{\bf G}[s]\Big]\ \, \lineup =\, {\bf G}[\half s\partial v
- v\partial s],\nonumber\\
\Big\{{\bf G}[s_1],{\bf G}[s_2]\Big\}\, \lineup =\, 
{\bf T}[2s_1s_2].\label{eq:superconf_alg}
\end{eqnarray}
To be careful about singularities at the midpoint, we regulate 
$\mathcal{G}$ and $\mathcal{L}^+$ by replacing
\begin{eqnarray} g(\xi)\lineup \rightarrow g(\lambda\xi),\ \ \ \ \ \ \,
 g(\xi)^\star\rightarrow g(\xi/\lambda)^\star, \nonumber\\
\ell(\xi)\lineup \rightarrow \ell(\lambda\xi),\ \ \ \ \ \ \ 
\ell(\xi)^\star\rightarrow \ell(\xi/\lambda)^\star.
\end{eqnarray}
Condition {\bf 3)} is satisfied in the limit $\lambda\to 1^-$. 
Another check is to expand the operators in modes
\begin{eqnarray}
\mathcal{L}\ \lineup =L_0+\frac{2}{3}L_2-\frac{2}{15}L_4+
\ \ \ \ \ \ \ \ .\, .\, .\ \ \ \ \ \ \ \ \ 
= L_0+2\sum_{n=1}^\infty \frac{(-1)^{n+1}}{4n^2-1}L_{2n},
\nonumber\\
\mathcal{L}^\star\, \lineup 
=L_0+\frac{2}{3}L_{-2}-\frac{2}{15}L_{-4}+
\ \ \ \ \ \ .\, .\, .\ \ \ \ \ \ \ \ 
= L_0+2\sum_{n=1}^\infty \frac{(-1)^{n+1}}{4n^2-1}L_{-2n},\nonumber\\
\widetilde{\mathcal{L}}^+ \lineup = \frac{\pi}{2}(L_1+L_{-1}),\nonumber\\
\mathcal{G}\ \lineup = \sqrt{\frac{\pi}{2}}\left(G_{-1/2}
+\frac{1}{2}G_{3/2}-\frac{1}{8}G_{7/2}+...\right)\ \,
= \sqrt{\frac{\pi}{2}}\sum_{n=0}^\infty
{1/2 \choose n} G_{2n-\frac{1}{2}}, \nonumber\\
\mathcal{G}^\star \lineup  = \sqrt{\frac{\pi}{2}}\left(G_{1/2}
+\frac{1}{2}G_{-3/2}-\frac{1}{8}G_{-7/2}+...\right)
= \sqrt{\frac{\pi}{2}}\sum_{n=0}^\infty
{1/2 \choose n} G_{\frac{1}{2}-2n},\end{eqnarray}
and calculate using the usual mode commutators of the 
superconformal algebra. Again we have found that {\bf 3)} is 
satisfied, and the infinite sums needed in the
computation are absolutely convergent.\footnote{One further subtlety is that 
the vanishing of left/right commutators is not always sufficient to 
guarantee that nonpolynomial combinations of left and right charges commute. 
Splitting $\mathcal{L}$ into left and right halves we find 
$[\mathcal{L}_L,\mathcal{L}_R]=0$, but $\mathcal{L}_L$ actually 
does not commute with $e^{-s\mathcal{L}_R}$. This is crucial for recovering 
closed string moduli in Schnabl gauge amplitudes\cite{KZ,boundary}. We have 
found no evidence for similar anomalies when splitting 
$\mathcal{L}^+$ and $\mathcal{G}$.} Given these and other checks, we believe 
that $\mathcal{L}^+,\mathcal{G}$ have a non-anomalous left/right 
decomposition. 

The operators $\mathcal{L},\mathcal{L}^\star,\tilde{\mathcal{L}}^+$ and 
$\mathcal{G},\mathcal{G}^\star$ form a super-Lie algebra with commutators,
\begin{eqnarray}[\mathcal{L},\mathcal{L}^\star]\lineup =\mathcal{L}^+,
\ \ \ \ \ \ \ \ 
[\mathcal{L},\widetilde{\mathcal{L}}^+]=\widetilde{\mathcal{L}}^+,
\ \ \ \ \ \ \ \ \ \ 
[\mathcal{L}^\star,\widetilde{\mathcal{L}}^+]=-\widetilde{\mathcal{L}}^+,
\nonumber\\
\ \{\mathcal{G},\mathcal{G}^\star\}\lineup =2\mathcal{L}^+,
\ \ \ \ \ \ \ 
\{\mathcal{G},\mathcal{G}\}=2\widetilde{\mathcal{L}}^+,
\ \ \ \ \ \ \ \ 
\{\mathcal{G}^\star,\mathcal{G}^\star\}=2\widetilde{\mathcal{L}}^+,
\nonumber\\
\ [\mathcal{L},\mathcal{G}]\lineup =\half \mathcal{G},
\ \ \ \ \ \ \ \ \ 
[\mathcal{L}^\star,\mathcal{G}] =-\half\mathcal{G},
\ \ \ \ \ \ \ \ 
\ [\widetilde{\mathcal{L}}^+,\mathcal{G}] =0,
\nonumber\\
\ [\mathcal{L},\mathcal{G}^\star]\lineup =\half\mathcal{G}^\star,
\ \ \ \ \ \ \,
[\mathcal{L}^\star,\mathcal{G}^\star]=-\half\mathcal{G}^\star,
\ \ \ \ \ \ \ 
[\widetilde{\mathcal{L}}^+,\mathcal{G}^\star]=0.
\end{eqnarray}
This can be thought of as a supersymmetric extension of the 
special projector algebra\cite{Schnabl,RZ}. Assuming 
\eq{KGop}, this algebra can be compactly summarized by the relations
\begin{equation}G^2=K\ \ \ \ \ [K,G]=0\ \ \ \ \ \half\mathcal{L}^-K=K
\ \ \ \ \ 
\half\mathcal{L}^-G = \frac{1}{2}G\end{equation}

\section{Splitting Charges and Midpoint Insertions}

When computing the action and gauge invariant overlap, we implicitly assumed
cyclicity of the vertices $\llangle\cdot\rrangle$ and 
$\llangle\cdot\rrangle_\mathcal{V}$. However, the presence of midpoint 
insertions makes this subtle. In the $K,B,c,G$ subalgebra, cyclicity of 
$\llangle\cdot\rrangle$ and $\llangle\cdot\rrangle_\mathcal{V}$ requires
\begin{equation}
[K,Y_{-2}]=0,\ \ \ \ \ \ \ [K,\mathcal{V}]=0,\end{equation}
and likewise for $B$ and $G$. (The cyclicity of $c$ appears 
unproblematic since the $c$ insertion is far from the midpoint.)
While the geometry of the Witten vertex appears to guarantee that midpoint 
insertions commute, this expectation fails in at least some 
examples\cite{Horowitz}.\footnote{A 
related question is whether the derivations $\mathcal{L}^-$ and 
$\mathcal{B}^-$ annihilate $\llangle\cdot\rrangle$ and 
$\llangle\cdot\rrangle_\mathcal{V}$. This can be shown along similar lines
to the argument presented here.}

To keep the discussion general, consider a string field $\Phi$ 
corresponding to a vertical line integral insertion of a primary 
$\phi(z)$ of weight $h>0$ in the cylinder coordinate frame:
\begin{equation}\Phi\ \rightarrow\ 
\int_{-i\infty}^{i\infty}\frac{dz}{2\pi i}\phi(z).\end{equation}
Explicitly we can write
\begin{equation}\Phi = {\bm \phi}_L|I\rangle,\ \ \ \mathrm{where}\ \ \ \ 
{\bm\phi}_L = \int_L \frac{d\xi}{2\pi i}\left(\sqrt{\frac{\pi}{2}}
\sqrt{1+\xi^2}\right)^{2(h-1)}\phi(\xi),\end{equation}
and the contour $L$ is over the positive half of the unit circle. Let 
$m=m(i)|I\rangle$ correspond to an insertion of a dimension zero primary $m(z)$
at the midpoint. Then we can show that $[\Phi,m]=0$ if and only if
\begin{equation}\lim_{\sigma\to\frac{\pi}{2}}[\,{\bm\phi}_L\,,
\,m(e^{i\sigma})\,]=0.
\label{eq:vanish_com}\end{equation}
Suppose $\phi(z)$ and $m(z)$ have an OPE of the form,
\begin{equation}\phi(z+w)m(z)\sim \sum_{n=1}^\infty \frac{1}{w^n}V_n(z),
\end{equation}
where $V_n(z)$ are local operators of dimension $h-n$. Computing the 
commutator \eq{vanish_com} we can prove the following:
\begin{claim2} The limit of the commutator 
\eq{vanish_com} vanishes if and only if one of the two following criteria are 
satisfied:
\begin{description}
\item{\bf a)} If $h\in \mathbb{Z}+\frac{1}{2}$, then $V_n(z)=0$ for all $n>h$.
\item{\bf b)} If $h\in \mathbb{Z}$, then $V_n(z)=0$ for all $n$ in the range 
$2h> n\geq h$.
\end{description}
\end{claim2}
\noindent In the current context, the role of $\Phi$ is played by $K,B$, and
$G$ and the role of $m$ is played by $Y_{-2}$ and $\mathcal{V}$. According
to the above claim, $K,B$ and $G$ commute with $Y_{-2}$ and $\mathcal{V}$ 
if and only if the OPEs between $T,G,b$, and $Y_{-2},\mathcal{V}$ take 
the following form:
\begin{eqnarray}T(z+w)Y_{-2}(z,\bar{z})\lineup\sim\mathcal{O}(w^{-1}),
\ \ \ \ \ \ 
T(z+w)\mathcal{V}(z,\bar{z})\sim\mathcal{O}(w^{-1}),\nonumber\\
G(z+w)Y_{-2}(z,\bar{z})\lineup\sim\mathcal{O}(w^{-1}),
\ \ \ \ \ \ 
G(z+w)\mathcal{V}(z,\bar{z})\sim\mathcal{O}(w^{-1}),\nonumber\\
b(z+w)Y_{-2}(z,\bar{z})\lineup\sim\mathcal{O}(w^{-1}),
\ \ \ \ \ \ 
b(z+w)\mathcal{V}(z,\bar{z})\sim\mathcal{O}(w^{-1}),
\end{eqnarray}
and likewise for the antiholomorphic currents $\tilde{T},\tilde{G},\tilde{b}$.
Let us assume that $Y_{-2}$ and $\mathcal{V}$ take the explicit forms given in 
\eq{Ym2} and \eq{V}. The OPEs with $T$ follow from the fact that $Y_{-2}$ 
and $\mathcal{V}$ are dimension $(0,0)$ primaries. The OPEs with $G$ 
follow from the fact that 
$Y_{-2}$ and $\mathcal{V}$ are superconformal primaries. Finally the OPEs
with $b$ follow from the fact that in the $bc$ CFT $Y_{-2}$ and $\mathcal{V}$
are proportional to $c\tilde{c}$, which produces only a single pole in the 
OPE with $b$. Therefore the vertices $\llangle\cdot\rrangle$ and 
$\llangle\cdot\rrangle_\mathcal{V}$ are expected to be cyclic when evaluated 
on fields in the $K,B,c,G$ subalgebra.

\section{Phantom Piece and Super-Wedge States}
\label{app:largeN}

In this appendix we prove that the phantom term \eq{phantom} can be described
by a super-wedge state \eq{ext_wedge} in the large $N$ limit. First we give 
an explicit definition of super-wedge states in the Fock space. Write
\begin{equation}e^{-\alpha K}e^{i\beta G} = f_1(\alpha,\beta)+iG\,
f_2(\alpha,\beta),
\end{equation}
where
\begin{eqnarray}
f_1(\alpha,\beta)\lineup =\Omega^\alpha\cos(\beta\sqrt{K}),\nonumber\\
f_2(\alpha,\beta)\lineup =\Omega^\alpha\frac{\sin(\beta\sqrt{K})}{\sqrt{K}}.
\end{eqnarray}
We can compute the Fock space coefficients of $(f_1,f_2)$ using the 
linear functional \eq{gen_f}:
\begin{eqnarray}L_{f_1}(x^h)\lineup
=\frac{1}{(h-1)!}\int_0^\infty dt\, t^{h-1}\cos(\beta\sqrt{t})e^{-(\alpha+1)t}
\nonumber\\
\lineup = \frac{1}{(1+\alpha)^h}
\,_1 F_1\left[h,\frac{1}{2},-\frac{\beta^2}{4(1+\alpha)}\right],
\label{eq:f1_fock}\\
L_{f_2}(x^h)\lineup
=\frac{1}{(h-1)!}\int_0^\infty dt\, t^{h-1}\frac{\sin(\beta\sqrt{t})}{\sqrt{t}}
e^{-(\alpha+1)t}\nonumber\\
\lineup = \frac{\beta}{(1+\alpha)^h}
\,_1 F_1\left[h,\frac{3}{2},-\frac{\beta^2}{4(1+\alpha)}\right],
\label{eq:f2_fock}
\end{eqnarray}
where $_1 F_1$ is the confluent hypergeometric function. 

Consider the states $(X_N,Y_N)$ appearing in the phantom piece through 
equation \eq{XY}. We can also define these states using the linear functional 
\eq{gen_f}:
\begin{eqnarray}L_{X_N}(x^h)\lineup = 
\frac{1}{(h-1)!}\int_0^\infty dt\, t^{h-1}
\frac{(1+ia\sqrt{t})^N+(1-ia\sqrt{t})^N}{2}e^{-(N+1)t},\nonumber\\
L_{Y_N}(x^h)\lineup = \frac{1}{(h-1)!}\int_0^\infty dt\,t^{h-1}
\frac{(1+ia\sqrt{t})^N-(1-ia\sqrt{t})^N}{2i aN\sqrt{t}}e^{-(N+1)t}.
\label{eq:XY_linfunc}\end{eqnarray}
To compute the large $N$ limit, substitute $s=(N+1)t$ in the integrand so that
for example
\begin{equation}L_{X_N}(x^h) = \frac{1}{2(h-1)!}\frac{1}{(N+1)^h}
\int_0^\infty ds\, s^{h-1}
\left[\left(1+ia\sqrt{\frac{s}{N+1}}\right)^N+\left(1-ia\sqrt{\frac{s}{N+1}}
\right)^N\right]e^{-s}.\end{equation}
Now approximate 
\begin{eqnarray}\left(1\pm ia\sqrt{\frac{s}{N+1}}\right)^N\lineup =
\exp\left[N\ln\left(1\pm ia\sqrt{\frac{s}{N+1}}\right)\right]\nonumber\\
\lineup = \exp\left[N\left(\pm ia\sqrt{\frac{s}{N}}+\frac{a^2}{2}\frac{s}{N}
+\mathcal{O}(N^{-1/2})...\right)\right]\nonumber\\
\lineup = e^{\pm ia\sqrt{Ns}}e^{a^2s/2}[1+\mathcal{O}(N^{-1/2})]
,\end{eqnarray}
so that 
\begin{eqnarray}L_{X_N}(x^h)\lineup = \frac{1}{(h-1)!}\frac{1}{N^h}
\int_0^\infty ds\, s^{h-1} \cos(a\sqrt{Ns}) e^{-(1-\frac{a^2}{2})s}
[1+\mathcal{O}(N^{-1/2})]\nonumber\\
\lineup = \frac{1}{N^h}\left(\frac{2}{2-a^2}\right)^h\, 
_1 F_1\left(h,\frac{1}{2},-\frac{1}{4}\frac{2}{2-a^2}a^2 N\right)
[1+\mathcal{O}(N^{-1/2})].
\end{eqnarray}
Similarly,
\begin{equation}L_{Y_N}(x^h) = \frac{1}{N^{h}}\left(\frac{2}{2-a^2}\right)^h
\,_1F_1\left(h,\frac{3}{2},-\frac{1}{4}\frac{2}{2-a^2}a^2N\right)
[1+\mathcal{O}(N^{-1/2})].
\end{equation}
Comparing with equations \eq{f1_fock} and \eq{f2_fock}, this precisely 
corresponds to the large $N$ behavior of the super-wedge state
$e^{iN a G}\Omega^{N(1-\frac{a^2}{2})}$, as claimed in equation \eq{Gwedge}.

To simplify the large $N$ limit further we use the asymptotic formula
\begin{equation}_1 F_1(a,b,z)=\frac{\Gamma(a)}{\Gamma(b-a)}e^{i\pi a}
\frac{1}{z^a}[1+\mathcal{O}(z^{-1})],\ \ \ \ \ (\mathrm{large}\ |z|,\ \ 
\mathrm{Re}(z)<0)\label{eq:hyper_ass}.\end{equation}
Thus,
\begin{eqnarray}L_{X_N}(x^h)\lineup = 
\frac{2(-1)^h}{(aN)^{2h}}\frac{(2h-1)!}{(h-1)!}
[1+\mathcal{O}(N^{-1/2})],\\
L_{Y_N}(x^h) \lineup = \frac{4(-1)^h}{(aN)^{2h}}\frac{(2h-3)!}{(h-2)!}
[1+\mathcal{O}(N^{-1/2})].
\end{eqnarray}
This agrees with the large $N$ behavior of the sums quoted in 
\eq{reg_limits}. We have verified this behavior numerically.

\section{Details of Energy Computation}
\label{app:energy}

In this appendix we give some details of the computation of the action for 
the simple half-brane solution. To avoid cluttered formulas, it is helpful to 
introduce the notation,
\begin{equation}(\Phi_1,\Phi_2) = \leftllangle 
\Phi_1\frac{1}{1+K}\Phi_2\frac{1}{1+K}\rightrrangle.\end{equation}
The kinetic term of the action can be expressed as the sum of 
two terms:
\begin{equation}\llangle \Psi Q\Psi\rrangle = -(1)+(2),\end{equation}
where
\begin{eqnarray}
(1) = \Big(cGBc,Q(cGBc)\Big),\ \ \ \ \ \ 
(2) = \Big(cGBcG, Q(cGBc)G\Big).
\end{eqnarray}
Now replace the $G$ insertions with supersymmetry variations $\delta$ acting
inside the vertex, following \eq{Gdelta}. This generates many terms, 
some of which vanish by $\phi$-momentum conservation or by $\mathcal{L}^-$ or
$\mathcal{B}^-$ invariance of the vertex. In the end the answer simplifies
to \begin{equation}(1) = -(cK,\gamma^2)+5(B\gamma^2, c\partial c)+2(\gamma,
\partial\gamma c)-4(cB\gamma, \partial\gamma c)-4(cB\gamma,\gamma Kc),
\end{equation}
and
\begin{eqnarray}(2) =\lineup -(cK,\gamma^2 K)-4(cB\gamma, c\partial\gamma K)
+2(cB\gamma, \partial c\gamma K)+(B\gamma^2,Kc\partial c)
-2(cB\gamma,K\partial \gamma c)\nonumber\\
\lineup+4(cB\gamma,K\gamma Kc)+(cB\gamma,\partial\gamma\partial c)
+2(cB\gamma,\partial^2\gamma c)-(cB\gamma,\gamma\partial^2 c)
-\frac{1}{2}(B\gamma^2,c\partial^2 c).\nonumber\\
\end{eqnarray}
We compute the inner products $(,)$ by mapping them to the appropriate 
correlation function on the cylinder, evaluating the correlator with 
\eq{corr}, and performing the Schwinger integrals. For $(1)$ the inner 
products turn out to be
\begin{eqnarray}(cK,\gamma^2) \lineup = \frac{2}{\pi^2},
\ \ \ \ \ \ \ \ 
(B\gamma^2,c\partial c) = \frac{1}{\pi^2},
\ \ \ \ \ 
(\gamma,\partial\gamma c)=-\frac{2}{\pi^2},
\nonumber\\
(cB\gamma,\partial\gamma c)\lineup = -\frac{1}{\pi^2},
\ \ \ \ 
(cB\gamma,\gamma Kc)= \frac{6}{\pi^4},\end{eqnarray}
giving
\begin{equation}(1) = -\frac{2}{\pi^2}+\frac{5}{\pi^2}-\frac{4}{\pi^2}
+\frac{4}{\pi^2}-\frac{24}{\pi^4}=\frac{3}{\pi^2}-\frac{24}{\pi^4}
.\end{equation}
For $(2)$ we have the inner products
\begin{eqnarray}
(cK,\gamma^2 K)\lineup =-\frac{1}{\pi^2},
\ \, \ \ \ \ 
(cB\gamma, c\partial\gamma K)=-\frac{1}{\pi^2},
\ \ \ \ \,
(cB\gamma, \partial c\gamma K)= -\frac{1}{2\pi^2},
\nonumber\\ 
(B\gamma^2,Kc\partial c)\lineup =-\frac{1}{2\pi^2},
\ \ \ \ 
(cB\gamma,K\partial \gamma c) = \frac{1}{\pi^2},
\ \ \ \ \ \ 
(cB\gamma,K\gamma Kc)=\frac{1}{2\pi^2}-\frac{6}{\pi^4},
\nonumber\\
(cB\gamma,\partial\gamma\partial c)\lineup =-\frac{1}{\pi^2},
\ \ \ \ \ \ \
(cB\gamma,\partial^2\gamma c)=\frac{1}{\pi^2},
\ \ \ \ \ \ \ \ \,
(cB\gamma,\gamma\partial^2 c) = 
(B\gamma^2,c\partial^2 c)= 0,\nonumber\\
\end{eqnarray}
giving
\begin{eqnarray}(2)\lineup = \frac{1}{\pi^2}+\frac{4}{\pi^2}-\frac{1}{\pi^2} 
-\frac{1}{2\pi^2}-\frac{2}{\pi^2}+\frac{2}{\pi^2}-\frac{24}{\pi^4}
-\frac{1}{\pi^2}+\frac{2}{\pi^2}+0+0\nonumber\\
\lineup = \frac{5}{\pi^2}-\frac{1}{2\pi^2}-\frac{24}{\pi^4}.
\end{eqnarray}
Adding things up
\begin{eqnarray}\llangle \Psi Q\Psi\rrangle \lineup = -(1)+(2) \nonumber\\
\lineup = -\frac{3}{\pi^2}
+\frac{24}{\pi^4}+\frac{5}{\pi^2}-\frac{1}{2\pi^2}-\frac{24}{\pi^4}\nonumber\\
\lineup = \frac{3}{2\pi^2}.
\end{eqnarray}
The energy is 
\begin{equation}E=-\frac{1}{6}\llangle \Psi Q\Psi\rrangle = -\frac{1}{4\pi^2},
\end{equation}
which is precisely $-1/2$ times the tension of the D-brane.

\section{Auxiliary Tachyon Coefficient}
\label{app:level}

In this appendix we compute the coefficient of the auxiliary tachyon state
$c_1|0\rangle$ for the Schnabl-like half-brane solution in the $L_0$ level
expansion. To achieve this we write the Schnabl-like solution in 
the form
\begin{eqnarray}\Psch \lineup = -\sum_{n=0}^\infty\psi_n' +\Gamma\nonumber\\
\lineup = 
-\sum_{n=0}^\infty\sum_{0\leq k\leq n/2}{n\choose 2k}a^{2k}
\left.\frac{d^{k+1}}{dr^{k+1}}\right|_{r=0}cB\Omega^{n+r}c(1+iaG)\Omega
\nonumber\\ \lineup\ \ \ \ \ 
-i\sum_{n=0}^\infty \sum_{0\leq k\leq \frac{n-1}{2}}{n\choose 2k+1}a^{2k+1}
\left.\frac{d^{k+1}}{dr^{k+1}}\right|_{r=0} cBG\Omega^{n+r}c(1+iaG)\Omega
\nonumber\\ \lineup\ \ \ \ \ 
+B\gamma^2(1+iaG)\Omega. \label{eq:Psch_Fock}
\end{eqnarray}
We can drop the phantom term since it vanishes in the Fock space. The states
inside the sums can be expressed using the operator formalism of 
Schnabl\cite{Schnabl,Schnabl_wedge}, which yields an expression for the 
solution in terms of a canonically ordered set of mode operators acting on 
the $SL(2,\mathbb{R})$ vacuum. Using \eq{reg_limits} one can argue that 
the infinite sums above converge for any coefficient in the Fock space as 
long as the parameter $a$ is restricted to the range 
$-\sqrt{2}\leq a\leq\sqrt{2}$.

Expanding \eq{Psch_Fock} in the Fock space we can extract the coefficient 
of the auxiliary tachyon. Define two functions
\begin{eqnarray}\phi_1(r) \lineup = \frac{1}{\pi X^2}\left[\frac{1}{\pi}
\cos^2\left(\frac{\pi}{2}X_+\right)\sin(\pi X_-)
- \frac{1}{\pi}\sin(\pi X_+)\cos^2\left(\frac{\pi}{2}X_-\right)\right.
\nonumber\\ \lineup \ \ \ \ \ \ \ \ \ \ \ \left.
-(X_+-1)\cos^2\left(\frac{\pi}{2}X_-\right)
+(X_-+1)\cos^2\left(\frac{\pi}{2}X_+\right)\right],\nonumber\\
\phi_2(r)\lineup = -\frac{d}{dr}\phi_1(r)+\frac{1}{X}
\left[-\frac{1}{2\pi} \sin(\pi X_+)\cos\left(\frac{\pi}{2}X_-\right)
-\frac{1}{2}(X_+ - 1)\cos\left(\frac{\pi}{2}X_-\right)\right.\nonumber\\
\lineup \ \ \ \ \ \ \ \ \ \ \ \ \ \ \ \ \ \ \ \ \ \ \ \ \ 
\left.-\frac{1}{4\pi}\sin(\pi X_+)\sin(\pi X_-)
-\frac{1}{2\pi}\cos^2\left(\frac{\pi}{2}X_+\right)(\cos(\pi X_-)+1)\right.
\nonumber\\ 
\lineup \ \ \ \ \ \ \ \ \ \ \ \ \ \ \ \ \ \ \ \ \ \ \ \ \ 
\left.-\frac{1}{4}(X_+-1)\sin(\pi X_-)\right],
\end{eqnarray}
where for short we have denoted
\begin{equation}X=\frac{2}{r+2},\ \ \ \ \ X_+ = \frac{r+1}{r+2},\ \ \ \ \ 
X_- =\frac{-r+1}{r+2}.\end{equation}
The auxiliary coefficient is then
\begin{eqnarray}\phi\lineup = 
\sum_{n=0}^\infty\left[-\sum_{0\leq k\leq n/2}{n\choose 2k}a^{2k}
\frac{d^{k+1}\phi_1(n)}{dn^{k+1}}
+\sum_{0\leq k\leq \frac{n-1}{2}}{n\choose 2k+1}a^{2k+2}
\frac{d^{k+1}\phi_2(n)}{dn^{k+1}}\right].\nonumber\\
\label{eq:aux_tach}\end{eqnarray}
Since $\phi_1$ and $\phi_2$ vanish as $1/r^3$ for large $r$, \eq{reg_limits} 
implies that the terms in the summand vanish as $1/n^8$ for sufficiently 
large $n$. We have checked this behavior numerically. Therefore 
\eq{aux_tach} is a convergent sum if $-\sqrt{2}\leq a\leq\sqrt{2}$. 
Unfortunately, the multiple derivatives of $\phi_1$ and $\phi_2$ make a 
direct numerical evaluation of \eq{aux_tach} very time-consuming. 
To evaluate \eq{aux_tach} with sufficient precision, we found it necessary 
to expand $\phi_1$ and $\phi_2$ in powers of $\frac{1}{r+2}$ out to 
$\frac{1}{(r+2)^{40}}$, which simplifies the numerical computation of 
derivatives. For $a=1$ we found the auxiliary tachyon coefficient to be
\begin{equation}\phi = -.0599156.\end{equation}
More interesting is the plot of the auxiliary tachyon coefficient as a 
function of $a$, shown in figure \ref{fig:aux_tach}. At $a=0$ the coefficient
corresponds to that of a tachyon vacuum solution, and has positive expectation
value, as we would expect from the usual picture of the cubic potential 
in bosonic string field theory. However, as $a$ becomes large, the 
expectation value becomes zero and even negative. This suggests that the 
negative energy of the half-brane solution is not principally due to the
condensation of the auxiliary tachyon. This is one way to see that the 
Schnabl-like solution must not satisfy the reality condition.

\begin{figure}
\begin{center}
\resizebox{3.2in}{1.9in}{\includegraphics{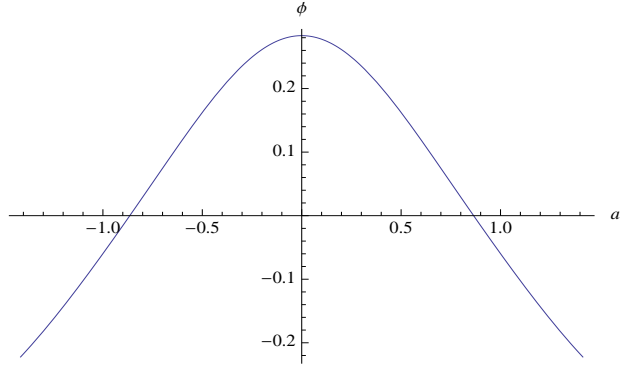}}
\end{center}
\caption{\label{fig:aux_tach} Coefficient of the auxiliary tachyon 
$c_1|0\rangle$ in the Schnabl-like solution, as a function of 
$a\in[-\sqrt{2},\sqrt{2}]$.} 
\end{figure}

We have also computed the coefficients for a few descendents of the auxiliary
tachyon. Let us denote coefficients of the states 
\begin{equation}(L_{-2})^n c_1|0\rangle,\ \ \ (L_{-4})^n c_1|0\rangle 
\end{equation}
by $x_n$ and $y_n$ respectively for $n\geq 1$. At $a=1$ we have found the 
explicit values
\begin{eqnarray}x_1\lineup = .067747,
\ \ \ \ \ \ \ \ \ \ \ 
 y_1 = -.019133,
\nonumber\\
x_2\lineup = .0060976,
\ \ \ \ \ \ \ \ \ \,
y_2 = .000064506,
\nonumber\\
x_3\lineup = -.000042514,
\ \ \ \ y_3 = 7.9488\times 10^{-7}.
\end{eqnarray}
We have computed $x_n$ and $y_n$ out to $n=60$ and found that they decay
significantly faster then the corresponding coefficients of 
$(L_{-2})^n|0\rangle$ and $(L_{-4})^n|0\rangle$ of the sliver state. We 
therefore believe that the Schnabl-like solution is a regular 
state in the $L_0$ level expansion.

\end{appendix}

\end{document}